\documentclass[12pt]{article}
\usepackage[utf8]{inputenc}
\usepackage[margin=1.25in]{geometry}
\usepackage{verbatim}
\usepackage{amsmath}
\usepackage{amsthm}
\usepackage{amssymb}
\usepackage{setspace}
\usepackage{graphicx}
\usepackage{bbm}
\usepackage[authoryear]{natbib}
\PassOptionsToPackage{normalem}{ulem}
\usepackage{ulem}
\usepackage{xcolor}
\usepackage{bm}
\usepackage{enumitem}
\usepackage{tabularx}
\usepackage{hyperref}
\definecolor{GmailBlue}{RGB}{42, 93, 176}
\hypersetup{
    colorlinks = true,
    allcolors = GmailBlue
}

\makeatletter

\newtheorem{theorem}{Theorem}[section]
\newtheorem{lemma}[theorem]{Lemma}
\newtheorem{cor}[theorem]{Corollary}

\newtheorem{prop}[theorem]{Proposition}

\theoremstyle{definition}
\newtheorem{definition}{Definition}
\newtheorem{assumption}{Assumption}
\newtheorem{remark}{Remark}

\begin{document}

\title{Inference in Regression Discontinuity Designs under Monotonicity\thanks{We thank our advisors Donald Andrews and Timothy Armstrong for continuous guidance. We are also grateful to participants at the Yale Econometrics Prospectus Lunch for insightful discussions.}}

\author{Koohyun Kwon\thanks{Department of Economics, Yale University, \texttt{koohyun.kwon@yale.edu}} \and Soonwoo Kwon\thanks{Department of Economics, Yale University, \texttt{soonwoo.kwon@yale.edu}} }
\date{November 23, 2020}

\maketitle

\begin{abstract}
  We provide an inference procedure for the sharp regression discontinuity
  design (RDD) under monotonicity, with possibly multiple running
  variables. Specifically, we consider the case where the true regression
  function is monotone with respect to (all or some of) the running variables
  and assumed to lie in a Lipschitz smoothness class. Such a monotonicity
  condition is natural in many empirical contexts, and the Lipschitz constant
  has an intuitive interpretation.  We propose a minimax two-sided
  confidence interval (CI) and an adaptive one-sided CI. For the two-sided CI,
  the researcher is required to choose a Lipschitz constant where she believes
  the true regression function to lie in. This is the only tuning parameter, and the resulting CI has uniform coverage and obtains
  the minimax optimal length. The one-sided CI can be constructed to maintain
  coverage over all monotone functions, providing maximum credibility in terms
  of the choice of the Lipschitz constant.  Moreover, the monotonicity makes it
  possible for the (excess) length of the CI to adapt to the true Lipschitz
  constant of the unknown regression function. Overall, the
  proposed procedures make it easy to see under what conditions on the
  underlying regression function the given estimates are significant, which can
  add more transparency to research using RDD methods.
\end{abstract}

\newpage

\section{Introduction}

Recently, there has been growing interest in honest and minimax optimal inference
methods in regression discontinuity designs
(\citealp{armstrong2018optimal}; \citealp{armstrong2020simple};
\citealp{imbens2019optimized}; \citealp{kolesar2018inference};
\citealp{noack2019bias}). This approach requires a researcher to specify a
function space where she believes the regression function to lie in, and the
inference procedures follow once this function space
is chosen. The methods proposed in the literature essentially use bounds on the
second derivatives to specify the function space. This is motivated by the
popularity of local linear regression methods in practice, which is often
justified by imposing local bounds on the second derivative of the regression
function. However, choosing a reasonable bound on the second derivative can be
difficult in practice.

We address this concern by considering the problem of conducting inference for
the sharp regression discontinuity (RD) parameter under monotonicity and a Lipschitz condition. Specifically, the regression function is assumed to be
monotone in all or some of the running variables, with a  bounded first derivative. Monotonicity naturally
arises in many regression discontinuity design (RDD) contexts, which is well documented by
\cite{babii2019isotonic}. The Lipschitz constant, or the bound on the first derivative, has an intuitive
interpretation, since this is a bound on how much the outcome can change if the
running variable is changed by a single unit. Hence, if the researcher reports
the inference results along with the Lipschitz constant used to run the proposed
procedure, it is easy to see under what (interpretable) conditions on the
regression function the researcher has obtained such results. We exploit the
combination of the monotonicity and the Lipschitz continuity restrictions to
construct a confidence interval (CI) which is efficient and maintains correct
coverage uniformly over a potentially large and more interpretable function
space.

We provide a minimax two-sided CI and an adaptive one-sided CI. For the
two-sided CI, the researcher is required to choose a bound on the first
derivative of the true regression function. The bound is the only tuning parameter, and the resulting
CI has uniform coverage and obtains the minimax optimal length over the class of
regression functions under consideration. Moreover, by exploiting monotonicity, the CI has a significantly shorter length than minimax CIs constructed under no such
shape restriction. To
our knowledge, this paper is the first to consider a minimax optimal procedure
when the regression function is assumed to be monotone.

The one-sided CI can be constructed to maintain coverage over all monotone
functions, providing maximum credibility in terms of the choice of the Lipschitz
constant. Due to monotonicity, the resulting CI still has finite excess length
as long as the true regression function has a bounded first derivative, where this bound
is allowed to be arbitrarily large and unknown.\footnote{This requires the
  regression function to be monotone in all the running variables.} This is in
contrast with minimax CIs constructed without the monotonicity condition, in
which case the length must be infinity to cover all functions. Moreover, our proposed one-sided CI adapts to the underlying smoothness class, resulting in a shorter CI when the true regression function has a smaller first derivative bound. This enables the
researcher to conduct non-conservative inference for the RD parameter, at the
same time maintaining honest coverage over a significantly larger space of
regression functions. The cost of such adaptation is that we can only construct either a one-sided lower or upper CI depending on the treatment allocation rule, but not both. We characterize this relationship between the treatment allocation rules and the direction of the adaptive one-sided CI we can construct.  

Our approach, especially the two-sided CI, is closely related to the literature
on honest inference in RDDs. By working with second derivative bounds, the
inference procedures in the literature are based on local linear regression estimators, which is in
line with the more conventional methods used in the RDD setting. However, it is
rather difficult to evaluate the validity of a second derivative bound specified
by a researcher. While it seems rather innocuous to ignore regression functions
with kinks, thus with infinite second derivatives, it is not clear how large or
how small the second derivative should be to be considered as ``too large'' or
``too small''. For this reason, the literature often recommends a sensitivity
analysis to strengthen the credibility of the inference result. However, the
credibility gain from the sensitivity analysis is limited when the smoothness
parameter is not easy to interpret. For example, it is hard to judge whether the
maximum value considered in the sensitivity analysis is large enough or not.

In contrast, the bound on the first derivative can be chosen based on more
straightforward empirical reasoning, since the first derivative has an intuitive
interpretation of a partial effect. For example, if an outcome variable $y$ and
a running variable $x$ are current and previous test scores, the class of
regression functions whose values increase no more than one standard deviation
of $y$ in response to 1/10 standard deviation increase in $x$ can be regarded as
reasonable. \cite{armstrong2020finite} take a similar approach of specifying
Lipschitz constants in their empirical application, in an inference problem for average treatment effects under a different setting. By imposing a bound on the first
derivative, our procedure is based on a Nadaraya--Watson type
estimator, with the boundary bias correctly accounted for.

The possibility of forming an adaptive one-sided CI in RDD settings under
monotonicity was first considered in \cite{armstrong2015adaptive} and
\cite{armstrong2018simple}. The difference is that these papers are concerned
with adapting to H\"older exponents $\beta \in (0, 1]$ while fixing the
Lipschitz constant. Here, we fix $\beta = 1$ and adapt to the Lipschitz
constant. 
When $\beta = 1$, what governs the performance of an adaptive CI is the size of
the constant multiplied to the rate of convergence, not the rate itself. This is
in contrast to the setting considered in \cite{armstrong2015adaptive} and
\cite{armstrong2018simple}, which primarily discuss rate-adaptation. In this
paper, we provide a procedure which makes the magnitude of the multiplying
constant reasonably small.

\cite{babii2019isotonic} also consider an
RDD setting with a monotone regression function. They introduce an inference
procedure based on an isotonic regression estimator, conveying a similar message
to ours that the monotonicity restriction can lead to a more efficient inference
procedure. Our approach differs from theirs in several key aspects: 1) we explicitly focus on maintaining uniform
coverage and optimizing the length of the CI, 2) we consider the Lipschitz class while they consider the H\"older class with exponent $\beta > 1/2$, and 3) our procedure can be used in settings with multiple running variables.\footnote{In fact, our procedure can be easily adapted to the case where $\beta \in (0,1]$, using the results provided in \cite{kwonkwon2019regpoint}. We focus on $\beta = 1$ mainly due to the interpretability of the Lipschitz constant, and because assuming bounded first derivatives seems rather innocuous in empirical contexts.}  

The general treatment of the dimension of the running variables in this paper
allows a researcher to use our procedure in the setting with more than one
running variables, referred to as the multi-score RDD by
\cite{cattaneo2020practical}. This setting has also been considered by
\cite{imbens2019optimized}. Our paper is the first to consider the space of
monotone regression functions in this setting, and the gain from monotonicity is
especially significant for the case with multiple running variables, as we show
later in this paper. We allow the regression function to be monotone with
respect to only some of the variables as well, which broadens the scope of application
of our procedure.

The rest of the paper is organized as follows. Section \ref{sec: setting} describes the setting and the form of optimal kernels and bandwidths under this setting. Section \ref{sec: minimax} introduces the minimax two-sided CI, and Section \ref{sec: adpt CI} the adaptive one-sided CI. Section \ref{sec: Monte Carlo} provides results from simulation studies to demonstrate the efficacy of the proposed procedures. Section \ref{sec: emp} revisits the empirical analysis by \cite{lee2008randomized}. Section \ref{sec: conc} concludes by discussing possible extensions.

\section{Setting}
\label{sec: setting}
We observe i.i.d. observations
$\left\{ \left(y_{i},x_{i} \right)\right\} _{i=1}^{n}$, where
$y_{i}\in \mathbb{R}$ is an outcome variable, and
$x_{i}\ \mathcal{\in X}\subset\mathbb{R}^{d}$ is either a scalar or a vector of
running variables. We take $\mathcal{X}$ to be a hyperrectangle in
$\mathbb{R}^{d}$, and let $ \mathcal{X}_{t} $ and $\mathcal{X}_{c}$ be connected
sets with nonempty interiors that form a partition of $\mathcal{X}$. The subscripts $t$ and $c$
correspond to ``treatment'' and ``control'' groups, respectively, throughout this paper. We write
$x_{i} \in \mathcal{X}_{t}$ and $x_{i} \in \mathcal{X}_{c}$ to indicate that
individual $i$ belongs to the treatment and the control groups,
respectively. When $d=1$, our setting corresponds to the standard sharp RD
design with a single cutoff point.

Let $\mathbbm{1}(\cdot)$ denote the indicator function. Then, our setting can be
written as a nonparametric regression model
\begin{align}
  \label{eq: RegModel}
  y_{i}=f(x_{i})+u_{i}, \hspace{15pt}
  f(x) =f_{t}(x) \mathbbm{1} \left\{ x \in \mathcal{X}_{t} \right\} +f_{c}(x) \mathbbm{1} \left\{ x \in \mathcal{X}_{c} \right\},
\end{align}
where the random variable $u_{i}$ is independent across $i$. Here, $f_{t}$ and
$f_{c}$ denote mean outcome functions for the treated and the control groups,
respectively. So $f$ corresponds to the mean outcome function for the observed
outcome, which we refer to as the ``regression function'' throughout the paper.

Our parameter of interest is a treatment effect parameter at a boundary point, $ \in \mathcal{B} := \partial \mathcal{X}_t \cap \partial \mathcal{X}_c$,
defined as
\begin{eqnarray*}
  L_{RD} f & := & \lim_{x \to x_0, x \in \mathcal{X}_{t}}f(x)-
                  \lim_{x \to x_0, x \in \mathcal{X}_{c}}f(x).
\end{eqnarray*}
When $d=1$, $L_{RD} f$ corresponds to the conventional sharp RD parameter. On
the other hand, when $d>1$, $L_{RD} f$ is the sharp RD parameter at a particular
cutoff point in $\mathcal{B}$. This type of parameter is also considered in
\citet{imbens2019optimized} and \citet{cattaneo2020practical} when they analyze
the RD design with multiple running variables. Without loss of generality, we
set $x_0 = 0$ since we can always relabel $\tilde{x}_i = x_i - x_0$.

\begin{remark}
[More than two treatment status]
Our framework can also handle the setting where individuals are assigned to more
than two treatments based on the values of the multiple running variables, which
was analyzed by \cite{papay2011extending}. For example, when $d = 2$ with two
cutoff points $c_1, c_2$, four different treatment status are possible depending
on whether $x_i \geq c_i$ for $i=1,2$. Let $\{\mathcal{X}_j\}_{j=1}^4$ be the
partition of $\mathcal{X}$, where $j = 1,...,4$ index the different treatment
status. Then, if we are interested in the treatment effect regarding the first
and the second treatments, we can let
$\widetilde{\mathcal{X}} = \mathcal{X}_1 \bigcup \mathcal{X}_2$, and apply our
method with $\widetilde{\mathcal{X}}$ instead of $\mathcal{X}$, with
$\mathcal{X}_t = \mathcal{X}_1$ and $\mathcal{X}_c = \mathcal{X}_2$.
\end{remark}
\subsection{Function space}

We consider the framework where a researcher is willing to specify some $C > 0$
such that
\begin{align}
  \label{eq: Lipschitz}
  \left|f_q(x)-f_q(z)\right|\leq C \cdot ||x - z|| \text{ for all }x,z\in \mathcal{X}_q,\text{ for each } q \in \{t,c\},
\end{align}
for some norm $|| \cdot ||$ on $\mathbb{R}^d$. In other words, the mean
potential outcome functions are Lipschitz continuous with a \emph{Lipschitz constant}
$C$, with respect to the norm $|| \cdot ||$. While the norm can be understood as
an absolute value when $d = 1$, the choice of the norm can give different
interpretations when $d > 1$. We allow a general class of norms, only requiring that $||x||$ is increasing in the absolute value of each element of $x$. We say the regression function
$f$ defined in (\ref{eq: RegModel}) has a Lipschitz constant $C$ if $f_t$ and
$f_c$ satisfy (\ref{eq: Lipschitz}).

In addition to the Lipschitz continuity, the researcher assumes that the mean
potential outcome functions are monotone with respect to the running
variables. Formally, letting $z_{(s)}$ denote the $s$th component of the running
variable $z \in \mathbb{R}^d$, and letting $\mathcal{V} \subseteq \{1,...,d\}$
be an index set for monotone variables, the researcher assumes that
\begin{align}
  \label{eq: monotone}
  f_q(x) \geq f_q(z) \text{ if } x_{(s)} \geq z_{(s)} \text{ for all } s \in \mathcal{V} \text{ and } x_{(t)} = z_{(t)} \text{ for all } t \notin \mathcal{V},
\end{align}
for all $q \in \{t, c\}$.  We use $\mathcal{F}(C)$ to denote the space of
functions on $\mathbb{R}^d$ which satisfy (\ref{eq: Lipschitz}) and (\ref{eq:
  monotone}) separately on $\mathcal{X}_t$ and $\mathcal{X}_c$.

We discuss the practical implications of our framework. First, there are
abundant settings where such a monotonicity condition is reasonable. See
Appendix \ref{sec: multi RD monotone} of our paper as well as Appendix A.3 of \cite{babii2019isotonic}
for examples of RDDs with monotone running variables.
One reason for the prevalence of monotonicity is due to the nature of policy
design. For example, students with lower test scores are assigned to summer
schools since policymakers are worried that students with lower test scores will
show lower academic achievement in the future---they believe the average future
academic achievement (an outcome variable) is monotone in current test scores
(running variables).

Next, our framework with Lipschitz continuity differs from the previous
approaches specifying a bound on the second derivative. For example,
\cite{armstrong2018optimal} consider the following locally smooth function
space:
\begin{align*}
  f_t, f_c \in \mathcal{F}_{T,p}
  :=
  \Big\{
  g:\ \Big\lvert
  g(x) - \sum\nolimits_{j = 0}^p g^{(j)}(0)x^j / j! 
  \Big\rvert \leq C |x|^p \ \forall x \in \mathcal{X}
 \Big\},
\end{align*}
and set $p = 2$ in their empirical application. Similarly,
\cite{imbens2019optimized} impose a global smoothness assumption that
$\left \lVert \nabla^2 f_q \right \rVert \leq C$ for $q \in \{t, c\}$, where
$\left \lVert \nabla^2 f_q \right \rVert$ denotes the operator norm of the
Hessian matrix. In contrary, we consider situations where a researcher has a
belief on the size of the first, rather than the second, derivative.

The main advantage of working with the first derivative is the interpretability
of the function space. The function space over which the coverage is uniform
should be easy to interpret, in the sense that the researcher herself or a
policymaker analyzing the inference procedure can evaluate whether the functions
not belonging to the function space can be safely disregarded as ``extreme''.

For this purpose, the size of the first derivative provides a reasonable
criterion. To be concrete, let us consider \cite{lee2008randomized} analyzing
the effect of the incumbency on the election outcome, where the outcome variable
is the difference in the percentage of votes between two parties in the current
election, and the running variable is the same quantity in the previous
election. Then, a mean potential outcome function whose maximum slope is as
large as $C = 50$ seems unreasonable---this roughly implies that a very small
increase in the vote percentage difference in an election, say 0.1\%, predicts a
large increase in the vote percentage difference in a consequent election,
5\%. Similarly, if a researcher presents a CI for the incumbency effect
parameter which has valid coverage over a space of functions with their first
derivatives bounded by $C = 0.5$, the policymaker evaluating the analysis might
find the function space too restrictive.

In comparison, it seems relatively more difficult to evaluate the validity of a
given bound on the second derivative. Previous papers have proposed heuristic
arguments to set the bound on the second derivative. For example,
\cite{armstrong2018optimal} choose the smoothness parameter so that the
reduction in the prediction MSE from using the true conditional mean rather than
its Taylor approximation is not too large. While this gives an alternative
interpretation of their smoothness constant, the prediction MSE does not have an
interpretation which can be connected to empirical examples being considered. \cite{imbens2019optimized} suggests
estimating the curvatures of $f_t$ and $f_c$ using quadratic functions and
multiplying a constant such as 2 or 3 to the estimates, but we can expect that
this procedure would not yield uniform coverage without further restrictions on
the function space, as pointed out in their paper. \cite{armstrong2020simple}
formally derive an additional condition on the function space which enables a
data-driven estimation of the smoothness parameter, but they warn that this
additional assumption may be difficult to justify. Instead of setting a single
bound, one may choose to conduct a sensitivity analysis, which is recommended by
\cite{armstrong2018optimal} and \cite{imbens2019optimized}. However, a
sensitivity analysis is more meaningful when the smoothness parameter the
researcher is varying has an intuitive meaning.

A possible drawback of having to specify a Lipschitz constant is that our
procedure does not ensure coverage when the mean potential functions are linear
or close to linear with slopes larger than the smoothness constant set a
priori. For the case of the minimax two-sided CI, we can view this as a
price we pay by maintaining validity for functions with arbitrarily large second
derivatives. On the other hand, our adaptive procedure
can be used to construct a one-sided CI which adapts to the degree of
Lipschitz smoothness. The adaptive procedure enables the researcher to set a very large value of the smoothness
parameter or even set it to infinity (so that the coverage is over all monotone
functions) and to obtain a shorter CI if the true regression function has a smaller first derivative bound.  This is possible due to the monotonicity assumption, which is
plausible in many RDD applications.

%
\begin{remark} [Lipschitz continuity under general dimension]
  \label{rem: multi-dim Lip}
  When $d > 1$, it can be more reasonable to assume that a researcher has a
  belief on the size of the partial derivative of the mean potential outcome
  functions. That is, there exist $C_1,...,C_d > 0$
  such that
  \begin{align}
    \label{eq: Lip partial}
    \left|f_q(x) - f_q(z)\right|\leq C_s |x_{(s)} - z_{(s)}| \text{ for all }x,z\in \mathcal{X} \text{ s.t. } x_{-s} = z_{-s},
  \end{align}
  for each $q \in \{c, t \}$ and $s \in \{1, \dots, d\}$. Here, $x_{-s}$ denotes the elements in $x \in \mathbb{R}^d$ excluding its
  $s$th component. It is easy to show that under (\ref{eq: Lip partial}), the
  original Lipschitz continuity assumption (\ref{eq: Lipschitz}) holds with
  $C = 1$ and $||\cdot||$ being a weighted $\ell_1$ norm on $\mathbb{R}^d$,
  $||z|| = \sum_{s = 1}^d C_s |z_{(s)}|$. Moreover, (\ref{eq: Lipschitz})
  holding with $C = 1$ and the weighted $\ell_1$ norm also implies (\ref{eq: Lip
    partial}). Therefore, a researcher assuming (\ref{eq: Lip partial}) can
  equivalently assume (\ref{eq: Lipschitz}) with this weighted $\ell_1$
  norm. This approach is also used in \cite{armstrong2020finite} in the context
  of inference for average treatment effects under unconfoundedness.
\end{remark}
\begin{remark}[RDDs without monotonicity]
  By taking $\mathcal{V} = \emptyset$, our procedure can be used to construct a
  minimax CI for the RD parameter without the monotonicity
  assumption. While other alternatives such as \cite{armstrong2018optimal} and
  \cite{imbens2019optimized} can be used to deal with this setting, our
  procedure is still useful to researchers who prefer imposing bounds on the
  first derivative rather than the second derivative, perhaps due to better
  interpretability.
\end{remark}
\subsection{Optimal kernel and bandwidths}

Our procedures depend on certain kernel functions and bandwidths that depend on
the Lipschitz parameter. We first introduce some notations. Given some
$z \in \mathbb{R}^{d}$ and the index set $\mathcal{V}$, we define
$\left(z\right)_{\mathcal{V}+}$ to be an element in $\mathbb{R}^d$ where its
$s$th element is given by
\begin{align*}
  \left(z\right)_{\mathcal{V}+(s)} :=
  \max\left\{ z_{(s)},0\right\}\mathbbm{1}(s\in \mathcal{V}) + 
  z_{(s)}\mathbbm{1}(s\notin \mathcal{V}), \hspace{5pt}
  s = 1,...,d.
\end{align*}
Similarly, we define $(z)_{\mathcal{V}-} := -\left(-z
\right)_{\mathcal{V}+}$. When $a$ is a scalar, we use square brackets $[a]_+$ to
denote $\max\{a, 0\}$. In addition, we define
$\sigma(x_i) := \text{Var}^{1/2}[u_i | x_i]$, and given an estimator
$\widehat{L}$ of $L_{RD}f$, we write $\text{bias}_f(\widehat{L})$ to denote
$E_f[\widehat{L} - L_{RD}f]$ and $\text{sd}(\widehat{L})$ to denote
$\text{Var}^{1/2}(\widehat{L})$.

The minimax procedure is based on the following kernel function
\begin{align}
  \label{eq: kernel (implementation)}
  K(z) := \left[1 - ( \left \lVert (z)_{\mathcal{V}+} \right \rVert + \left \lVert (z)_{\mathcal{V}-} \right \rVert ) \right]_{+}.
\end{align}
In the adaptive procedure, different bandwidths are used for each coordinates,
depending on the signs of the coordinates. To make the use of different
bandwidths clear, for $h = (h_1, h_2) \in \mathbb{R}^2$, we define
\begin{align*}
  K(z, h) := \left[1 - ( \left \lVert (z / h_1)_{\mathcal{V}+} \right \rVert + \left \lVert (z / h_2)_{\mathcal{V}-} \right \rVert ) \right]_{+}.
\end{align*}

The optimal bandwidths used in estimation is based on the following two
functions $\omega_t(\delta_t; C_1, C_2)$ and $\omega_c(\delta_c; C_1, C_2)$,
defined to be the solutions to the following equation:
\begin{align*}
  \sum_{x_{i} \in \mathcal{X}_{q}} \left[ \omega_q(\delta_q; C_1, C_2) - C_1 \left\lVert (x_{i})_{\mathcal{V}+} \right\rVert - C_2 \left\lVert (x_{i})_{\mathcal{V}-} \right\rVert \right]_{+}^{2} / \sigma^{2}(x_{i}) 
  = \delta_q^{2}, \text{ for each } q \in \{t,q\},
\end{align*}
given a pair of non-negative numbers $(\delta_t, \delta_c)$, and
$\sigma^2(x) := \text{Var}[u_i | x_i = x]$. Moreover, given some scalar
$\delta \geq 0$, we define $\delta_t^*(\delta; C_1, C_2)$ and
$\delta_c^*(\delta; C_2, C_1)$ to be solutions to
\begin{align*}
  \sup_{\delta_{t}^* \geq 0, \delta_{c}^* \geq 0,
  (\delta_{t}^*)^{2} + (\delta_{c}^*)^{2} = \delta^{2}}
  \omega_{t} \left( \delta_{t}^*; C_1, C_2 \right) + 
  \omega_{c} \left( \delta_{c}^*; C_2, C_1 \right).
\end{align*}
Based on these definitions, we introduce the following shorthand notations to be
used throughout this paper; for $q \in \{t, c\}$, $\delta \geq 0$,
$(C_1, C_2) \in \mathbb{R}^2_+$, $h \in \mathbb{R}^2_+$, we define
\begin{align*}
  \omega_q^*(\delta; C_1, C_2) &:= \omega_q(\delta_q^*(\delta; C_1, C_2); C_1, C_2), \\
  a_q(h; C_1, C_2)  &:=
                      \frac{1}{2} \cdot \frac{\sum_{x_i \in \mathcal{X}_q} K(x_i, h) \left[ C_1 \left \lVert (x_i)_{\mathcal{V}+} \right \rVert
                      - C_2 \left \lVert 
                      (x_i)_{\mathcal{V}-} \right \rVert 
                      \right] / \sigma^2(x_i)}{\sum_{x_i \in \mathcal{X}_q} K(x_i, h) / \sigma^2(x_i)}, \\
  s(\delta, h; C_1, C_2, \sigma^2(\cdot)) &:= 
                                            \frac{\delta/\omega_{t}^* \left(\delta;
                                            C_1, C_2 \right)}
                                            {\sum_{x_{i} \in \mathcal{X}_{t}}
                                            K(x_i, h)/ \sigma^2(x_i)}.
\end{align*}
Moreover, for $q \in \{t, c\}$, $\delta \geq 0$, and $C_1, h_1 \in \mathbb{R}$,
we define
\begin{align*}
  \omega_q^*(\delta; C_1) &:= \omega_q^*(\delta; C_1, C_1), \\
  a_q(h_1; C_1) &:= a_q(h_1, h_1; C_1, C_1), \\
  s(\delta, h_1; C_1, \sigma^2(\cdot)) &:=
                                         s(\delta, h_1, h_1; C_1, C_1, \sigma^2(\cdot)).
\end{align*}

The forms of the optimal kernel and the bandwidth presented in this section
result from solving \emph{modulus of continuity} problems considered in
\cite{donoho1991geometrizing} and \cite{Donoho1994} in the context of minimax
optimal inference, and in \cite{cai2004adaptation} and
\cite{armstrong2018optimal} in the context of the adaptive inference. While we
only make the connection specifically to our setting when we prove the validity
of our procedure in Appendix \ref{app: lem and proof}, interested readers may
refer to the aforementioned papers for a more general discussion.

\section{Minimax Two-sided CI}
\label{sec: minimax}

We first consider the case where the researcher is comfortable specifying a
Lipschitz constant and/or the empirical context requires a two-sided CI. In this
case, we recommend a minimax affine CI, which is the CI whose worst-case
expected length is the shortest among all CIs based on affine estimators
(\citealp{Donoho1994}). We refer to such a CI as the minimax CI for brevity.\footnote{\cite{Donoho1994} shows focusing on affine
  estimators is reasonable in the sense that the gain from considering non-affine
  estimators is limited.}

The minimax CI is constructed based on an affine estimator
$\widehat{L}^\text{mm} := a + \sum_{i = 1}^n w_i y_i$ with non-negative weights
$(w_i)_{i = 1}^n$ and half-length $\chi^\text{mm}$, i.e.,
\begin{align*}
  CI^\text{mm} = \left[
  \widehat{L}^\text{mm} - \chi^\text{mm},
  \widehat{L}^\text{mm} + \chi^\text{mm}
  \right].
\end{align*}
The half-length $\chi^\text{mm}$ is non-random and calibrated to maintain
correct coverage uniformly over the function space $\mathcal{F}(C)$:
\begin{align}
  \label{eq: honesty (implementation)}
  \inf_{f \in \mathcal{F}(C)} P_f \left( L_{RD}f \in CI^\text{mm} \right) \geq 1 - \alpha.
\end{align}
Note that
\begin{align*}
  P_f \left( L_{RD}f \in CI^\text{mm} \right)
  =
  P_f \left( \left| \frac{
  \widehat{L}^\text{mm} - L_{RD}f
  }{\text{sd}( \widehat{L}^\text{mm} )}\right| \leq  \frac{\chi^\text{mm}}{\text{sd}( \widehat{L}^\text{mm} )} \right).
\end{align*}
Under the assumption that the error term $u_i$ has a Gaussian distribution, the random variable
$({ \widehat{L}^\text{mm} - L_{RD}f })/{\text{sd}( \widehat{L}^\text{mm} )}$ is
normally distributed with mean equal to
$\text{bias}_f(\widehat{L}^\text{mm}) / {\text{sd}( \widehat{L}^\text{mm})}$ and
with unit variance. Hence, the quantiles of the random variable
$| ({ \widehat{L}^\text{mm} - L_{RD}f }) / {\text{sd}( \widehat{L}^\text{mm} )}
|$ is maximized over $f \in \mathcal{F}(C)$ when
$| \text{bias}_f(\widehat{L}^\text{mm}) |$ is the largest. Therefore, if we
define $\text{cv}_\alpha (t)$ to be $1 - \alpha$ quantile of $|Z|$, with
$Z \sim N(t, 1)$, the smallest possible value of $\chi^\text{mm}$ that
guarantees the coverage requirement (\ref{eq: honesty (implementation)}) is
given by
\begin{align}
  \label{eq: chi^mm general}
  \chi^\text{mm} = 
  \text{cv}_\alpha \left(
  \frac{\sup_{f \in \mathcal{F}(C)}
  \lvert \text{bias}_f(\widehat{L}^\text{mm}) \rvert}{\text{sd}( \widehat{L}^\text{mm})}
  \right) \text{sd}( \widehat{L}^\text{mm} ).
\end{align}

It remains to derive the form of the estimator $\widehat{L}^\text{mm}$ such that
$\chi^\text{mm}$ is minimized. We take $\widehat{L}^\text{mm}$ to be the
difference between two re-centered kernel regression estimators, say
$\widehat{L}^\text{mm}_t - \widehat{L}^\text{mm}_c$, where
\begin{align}
  \label{eq: Lhat_q mm def}
  \widehat{L}^\text{mm}_q :=   \frac{\sum_{x_i \in \mathcal{X}_q} K(x_i/h_q) y_i / \sigma^2(x_i)}{\sum_{x_i \in \mathcal{X}_q} K(x_i/h_q) / \sigma^2(x_i)} - a_q, \text{ for each } q \in \{t, c\},
\end{align}
for the kernel function $K(\cdot)$ defined in (\ref{eq: kernel
  (implementation)}). Note that $\widehat{L}^\text{mm}_t$ and
$\widehat{L}^\text{mm}_c$ correspond to estimators of $f_t(0)$ and $f_c(0)$,
respectively. Regarding the form of the optimal kernel function $K$, $K(z)$ is
the usual triangular kernel when $d = 1$, whose optimality is discussed in
\cite{Donoho1994} and \cite{armstrong2020simple}. Here, we derive the optimal
kernel for multi-dimensional cases as well, for any given norm and under partial
or full monotonicity.

A notable difference from the previous inference methods in RDDs is that the
estimator $\widehat{L}_q^{\text{mm}}$ is a Nadaraya-Watson estimator instead of
a local linear estimator. This difference naturally arises because we work under
the assumption of bounded first derivatives. In general, the local linear
estimator is preferred due to the well-known issue of bias at the boundary for
Nadaraya-Watson type estimators. In the context of honest inference, however,
the worst-case bias is explicitly corrected for.

For the optimal choices of bandwidths $(h_t, h_c)$ and the centering terms
$(a_t, a_c)$, we can show that the minimax two-sided CI can be obtained
by taking
\begin{align}
  \label{eq: h_q (implementation)}
  h_q = h_q(\delta) = {\omega_{q}^* \left( \delta; C \right)} / {C} \text{ for each } q \in \{t,c\},
\end{align}
for a suitable choice of $\delta \geq 0$, by applying the result of
\cite{Donoho1994} to our setting. Next, from the form of (\ref{eq: chi^mm
  general}), the centering terms should be chosen such that
\begin{align*}
  \sup_{f \in \mathcal{F}(C)}
  \text{bias}_f(\widehat{L}^\text{mm}) 
  = 
  - \inf_{f \in \mathcal{F}(C)}
  \text{bias}_f(\widehat{L}^\text{mm}),
\end{align*}
to minimize $\chi^\text{mm}$, i.e., the worst-case negative and positive biases
should be balanced. We can show that this can be achieved if we choose
\begin{align}
  \label{eq: a_q (implementation)}
  a_q  :=
  a_q(h_q; C), \text{ for each } q \in \{t,c\}.
\end{align}
Note that this quantity depends on $\delta$ when $(h_t, h_c)$ depends on
$\delta$.

Now, let $\widehat{L}^\text{mm}(\delta)$ be the above kernel regression
estimator with the bandwidth $(h_t(\delta), h_c(\delta))$ and centering terms
$(a_t, a_c)$ as defined above. Under these choices of bandwidths and centering
terms, we have
\begin{align}
  \begin{aligned}
    \text{sd}( \widehat{L}^\text{mm}(\delta)
    )& = s(\delta, h_{t}(\delta); C, \sigma^2(\cdot)), \label{eq: sd and sup bias (implementation)} \\
    \sup_{f \in \mathcal{F}(C)} \text{bias}_f(\widehat{L}^\text{mm}(\delta) ) &= -
    \inf_{f \in \mathcal{F}(C)} \text{bias}_f(\widehat{L}^\text{mm}(\delta))  \\
    &=\frac{1}{2} \left( C(h_{t}(\delta) + h_{c}(\delta)) - \delta \, \text{sd}(
      \widehat{L}^\text{mm}(\delta))\right).
  \end{aligned}
\end{align}
Then, we choose the optimal value of $\delta$ by plugging in \eqref{eq: sd and
  sup bias (implementation)} into \eqref{eq: chi^mm general}, and calculate the
value of $\delta$ that minimizes the half-length $\chi^\text{mm}$, say
$\delta^*$, yielding the value of the shortest half-length. Plugging
$\delta = \delta^*$ into the bandwidth and centering term formulas also yields
the form of the estimator corresponding to this half-length. Procedure 1
summarizes our discussion on the construction of the minimax CI.
\begin{table}
  \begin{tabularx}{\textwidth}{X}
    \hline
    \textbf{Procedure 1} \emph{Minimax Affine CI}. \\
    \hline
    \begin{minipage}[t]{\linewidth}
      \begin{enumerate}
      \item Choose a value $C$ such that the Lipschitz continuity in (\ref{eq:
          Lipschitz}) is satisfied, with a suitable choice of a norm $||\cdot||$
        on $\mathbb{R}^d$ when $d > 1$.
      \item Calculate the form of the estimator using (\ref{eq: Lhat_q mm def})
        with the bandwidth and the centering term given by (\ref{eq: h_q
          (implementation)}) and (\ref{eq: a_q (implementation)}), as functions
        of $\delta$.
      \item Using \eqref{eq: sd and sup bias (implementation)}, find the value
        of $\delta$ which minimizes the half-length (\ref{eq: chi^mm general}),
        say $\delta^*$.
      \item Calculate the value of the estimator and the half-length by plugging
        in $\delta^*$, which gives the final form of the CI.
      \end{enumerate}
    \end{minipage} \\
    \\
    \hline
  \end{tabularx}
\end{table}

The following is the main theoretical result regarding the minimax
procedure. While we consider an idealized setting with Gaussian errors and
known conditional variances, such exact finite sample results can be translated
into asymptotic results under non-Gaussian errors with unknown variances
following similar arguments given by \cite{armstrong2018optimal_supp}. In
Section \ref{sec: Monte Carlo}, we discuss in more detail how to plug in
consistent estimators of the conditional variances.
\begin{assumption}
  \label{assn: Gaussian temp}
  $\{x_i\}_{i = 1}^n$ is nonrandom and
  $u_i \sim N(0, \sigma^2(x_i)), \text{ where } \sigma^2(\cdot)$ is known.
\end{assumption}
\begin{theorem}
  \label{thm: minimax}
  Under Assumption \ref{assn: Gaussian temp}, we have
  \begin{align*}
    \inf_{f \in \mathcal{F}(C)} P_{f}\left(
    L_{RD}f \in CI^{\text{\emph{mm}}}
    \right) = 1 - \alpha.
  \end{align*}
  Moreover, $CI^{\text{\emph{mm}}}$ is the shortest among all (fixed-length)
  affine CIs with uniform coverage.
\end{theorem}

From (\ref{eq: h_q (implementation)}), we can see that there is a one-to-one
relationship between the size of the bandwidth and the Lipschitz constant $C$
chosen by a researcher. So choosing $C$ is not necessarily an additional burden
to the researcher if a bandwidth has to be chosen anyway. While there also exist
various data-driven bandwidth choice methods, our way of choosing the bandwidth
makes it clear the relationship between the bandwidth and the function space
over which the resulting CI has uniform coverage, at the same time achieving the
minimax optimal length.

It is useful to discuss the case with $d = 1$ to illustrate the role of the
monotonicity restriction in the minimax optimal inference. Intuitively, under
monotonicity, we do not have to worry about the bias caused by functions with
negative slopes, so it is optimal to use a larger bandwidth than in the
case without monotonicity in order to reduce the standard error. Using our kernel
function and bandwidth formulas above, we can calculate how much larger the
bandwidth should be under monotonicity. When $\mathcal{V} = \{1\}$, the kernel
function in (\ref{eq: kernel (implementation)}) is given by
$K_1(z) = [1 - |z|]_+$, while when $\mathcal{V} = \emptyset$, the kernel
function is given by $K_0(z) = [1 - 2|z|]_+ = K_1(z/(1/2))$. Therefore, if we
fix the kernel function to be $K(z) = K_1(z)$, the bandwidth ratio between the
one under $\mathcal{V} = \{1\}$ and the one under $\mathcal{V} = \emptyset$ is
given by
\begin{align*}
  \frac{2 \omega_q^*(\delta; C, \{1\})}{\omega_q^*(\delta; C, \emptyset)}, \text{ for each } q \in \{t,c\},\ \delta \geq 0,
\end{align*}
where $\omega_q^*(\delta; C, \mathcal{V})$ denotes the value of
$\omega_q^*(\delta; C)$ when the monotonicity restriction holds for the index
set $\mathcal{V}$. Following \cite{kwonkwon2019regpoint}, we can show that the
above quantity is approximately $2^{2/3} \approx 1.6$ for large $n$. So if we
believe the mean potential outcome functions are monotone, it is optimal to use
a bandwidth about 60\% larger than what should be used without monotonicity.

\begin{figure}[t!]
  \centering \includegraphics[width = 0.9\textwidth]{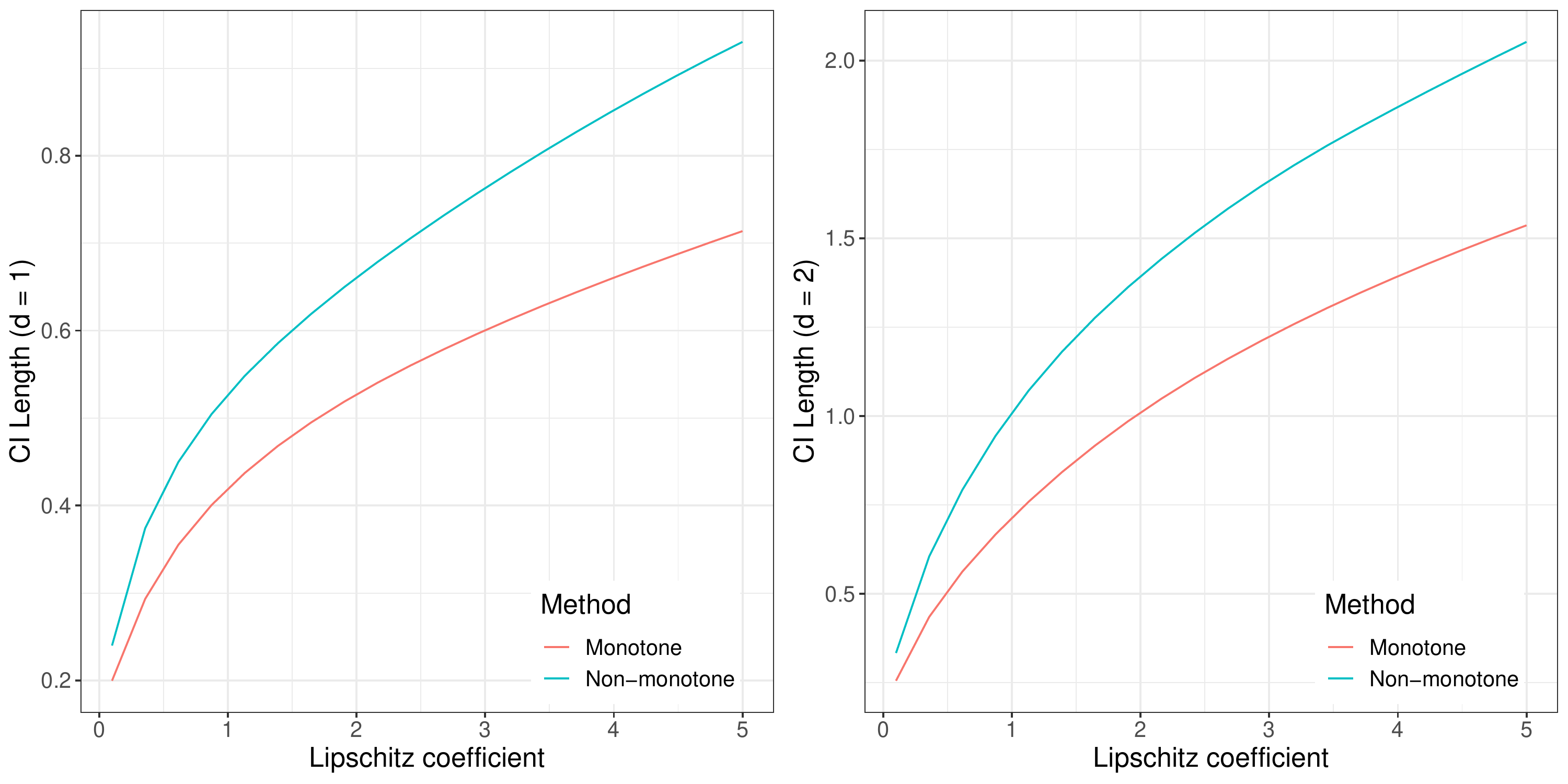}
  \caption{Comparison of the minimax lengths with and without monotonicity. For
    the design of the running variable(s), 500 observations were generated from
    the uniform distribution over $[-1, 1]^d$, for $d = 1, 2$. When $d = 2$,
    we set $\mathcal{V} = \{1, 2\}$ and used the $l_1$ norm. The lengths were
    calculated for $\sigma(x) = 1$.}
  \label{fig: gain mono}
\end{figure}

By a similar argument, the length of the minimax CI becomes shorter when
we only consider the space of monotone functions. Since the length of the CI is
a fixed quantity, we can easily compare its length under the shape restriction
to the one without it. Figure \ref{fig: gain mono} considers the cases with
$d = 1$ and $2$, for the treatment design where individuals are treated when the
values of all the running variables are negative. The CI which does not take
into account monotonicity can be as much as 30\% longer for $d = 1$ and 40\%
longer for $d = 2$. Considering the prevalence of RDDs with monotone regression
functions, this efficiency gain demonstrates the importance of incorporating the
shape restriction when constructing minimax CIs.
\section{Adaptive One-sided CI}
\label{sec: adpt CI}

A prominent feature of the minimax two-sided is that it
is a fixed-length CI, in the sense that its length is determined before observing the
realization of the outcome $\left\{ y_i \right\}_{i = 1}^n$. The length of the
CI depends on the Lipschitz constant $C$ chosen by the researcher, and the
larger the value of $C$, the longer the CI is. Hence, the minimax CI may become too
wide to use in the case where a researcher desires to strengthen the credibility
of the inference result by setting a conservative value of the Lipschitz
constant. In an extreme case where a researcher wishes to set $C = \infty$, the
CI is necessarily the entire real line, providing no information.

To deal with this issue, we provide a method to construct an adaptive one-sided CI. The CI can be made to maintain coverage over a function space with a large Lipschitz constant $C$, even allowing for $C = \infty$. Moreover, as long as the true regression function lies in a smoother class, the length of the CI does not depend on $C$ when $d = 1$ and nearly so when $d > 1$. In Section \ref{subsec: direction}, we discuss this property in more detail, and provide a simple condition on the relationship between the treatment allocation rule and the direction of monotonicity under which such a property holds. This property ensures that a researcher can strengthen the credibility of the inference result by considering a large function space without ending up with an uninformative CI. 

Furthermore, the one-sided CI can be made to utilize the information about the smoothness of the unknown regression function $f$ contained in the observed outcome $\{y_i\}_{i = 1}^n$. By using this
information, the length of the CI is adjusted accordingly, unlike the minimax optimal CI whose length is fixed regardless of the values of the observed outcome. The CIs possessing this type of property are called \emph{adaptive} CIs,
following \cite{cai2004adaptation}. This property further improves the usefulness of the proposed one-sided CI; the length of the CI is not only (nearly) independent of $C$ but shrinks when the true regression function has a smaller first derivative bound.

We focus our discussion on the construction of adaptive lower CIs and related treatment allocation rules. In many RDD applications, researchers are often interested in how significantly larger than 0 the true treatment effect is. In this context, a lower one-sided CI provides useful information. The upper CI can be dealt with in an analogous manner, but requires different (or ``opposite'') treatment allocation rules. 
 Lastly, the ``length'' of a one-sided CI refers to the distance between the true parameter $L_{RD}f$ and its endpoint, referred to as the ``excess length'' in \cite{armstrong2018optimal}.
\subsection{Conditions on treatment allocation}
\label{subsec: direction}

In this section, we describe in more detail the treatment
allocation rules under which it is possible to construct a lower CI with uniform coverage over $\mathcal{F}(C)$, but with length that does not necessarily increase with $C$. Specifically, given a smaller Lipschitz
constant $C'$ such that $0 \leq C' < C$, we ask when it is possible to construct a lower CI whose worst-case expected length over $\mathcal{F}(C')$ does not grow
with $C$. This property ensures that more credibility does not necessarily lead to wider CIs.

To give an intuition, we first describe the argument in a simple setting, where
there is a single running variable with the cutoff point $x = 0$. Consider a
lower CI $[ \widehat{L} - \chi^L, \infty )$ constructed by subtracting a
constant $\chi^{L}$ from a linear estimator $\widehat{L}$. In order to
maintain uniform coverage over a function space $\mathcal{F}(C)$, we must have
\begin{align*}
  \chi^L = \sup_{f \in \mathcal{F}(C)} \text{bias}_f(\widehat{L}) + \text{sd}(\widehat{L})\cdot z_{1 - \alpha}.
\end{align*}
Note that the estimator $\widehat{L}$ we consider is given as the
difference between estimators for $f_t(0)$ and $f_c(0)$, say
$\hat{f}_t(0) - \hat{f}_c(0)$.

Now, suppose individuals receive some treatment if and only if $x_i < 0$. A key
property under this design is that $\hat{f}_t(0)$ always has a negative bias if
$f_t(x)$ is increasing, since $\hat{f}_t(0)$ is calculated only using
observations with $x_i < 0$. Similarly, $\hat{f}_c(0)$ always has positive
bias when $f_c(x)$ is increasing. Therefore, the bias of $\widehat{L}$ over
$f \in \mathcal{F}(C)$ is always negative, regardless of the
value of $C$ specified by the researcher. Hence, a one-sided lower CI that maintains uniform coverage over $\mathcal{F}(C)$ can be formed to be \emph{independent} of $C$. We can also easily see that this argument no longer holds when the individual $i$ is treated if and only if $x_i \geq 0$, in which case the maximum bias over $f \in \mathcal{F}(C)$ increases with $C$.

We now state this idea formally, including the case where $d > 1$. Given some $\widetilde{C} > 0$, let
$\mathcal{L}_\alpha(\widetilde{C})$ denote the set of lower CIs (as functions of
$\{y_i\}_{i = 1}^n$ and $\{x_i\}_{i = 1}^n$) with uniform coverage probability
$1 - \alpha$ over $\mathcal{F}(\widetilde{C})$, i.e,
\begin{align*}
  CI \in \mathcal{L}_\alpha(\widetilde{C})
  \Longleftrightarrow
  \inf_{f \in \mathcal{F}(\widetilde{C})}
  P(L_{RD}f \in CI) \geq 1 - \alpha.
\end{align*}
Now, consider the following quantity:
\begin{align}
  \label{eq: one-sided mm problem}
  \ell(C'; C) := 
  \inf_{[\hat{c}, \infty) \in \mathcal{L}_\alpha(C)}
  \sup_{f \in \mathcal{F}(C')}
  E[L_{RD}f - \hat{c}],
\end{align}
where $C' \leq C$.  This is the worst-case expected length of a CI 1) which has
correct uniform coverage over $\mathcal{F}(C)$ and 2) whose worst-case expected
length is the smallest over $\mathcal{F}(C')$. \cite{armstrong2018optimal}
calculate this quantity in terms of the modulus of continuity defined in
Appendix \ref{app: lem and proof}. The question we ask here is under what
conditions the quantity $\ell(C'; C)$ can be viewed as independent of $C$. In
other words, we characterize when it is possible to construct a one-sided lower
CI whose length does not grow with $C$ when the regression function belongs to a
smoother function space.
\begin{lemma}
  \label{lem: lower adpt cond}
  Suppose both $\mathcal{X}_{t}$ and $\mathcal{X}_{c}$ are non-empty, and
  Assumption \ref{assn: Gaussian temp} holds. Then, given a pair of Lipschitz
  constants $(C',C)$ such that $C' < C$, there exists some constant $A(C')$,
  independent of $C$, such that
  \begin{align*}
    \ell(C'; C) \leq A(C'),
  \end{align*}
  if and only if the followings hold: 1) there exists $x_i \in \mathcal{X}_t$
  such that $x_i \in (-\infty,0]^{d}$, 2) there exists $x_i \in \mathcal{X}_c$
  such that $x_i \in [0,\infty)^{d}$, and 3) the regression function is fully
  monotone, i.e., $\mathcal{V} = \{1,...,d \}$.
\end{lemma}
Under the conditions provided in Lemma \ref{lem: lower adpt cond}, the researcher can construct a one-sided CI whose worst-case expected
length is not too large if the true regression function belongs to a smoother
space $\mathcal{F}(C')$, while maintaining a stringent coverage requirement by specifying a large
$C$. When $d = 1$, it is possible to take $A(C') = \ell(C'; C')$ with the inequality in Lemma \ref{lem: lower adpt cond} holding with equality, as suggested by the intuition discussed above. That is, $C$ does not affect the size of $\ell(C'; C)$ at all when $d = 1$. On the other hand, when $d > 1$, the same property does not hold unless we have observations only over $[0,\infty)^{d}$ and $(-\infty,0]^{d}$, which is usually not the case in practice. Therefore, specifying a larger $C$ translates into a longer CI (by giving less weights to observations outside $[0,\infty)^{d} \cup (-\infty,0]^{d}$) when $d > 1$. However, there is an upper bound on how much this length can grow with $C$. This upper bound is independent of $C$, and thus the worst-case length of the CI is \textit{nearly independent} of $C$.

In words, the conditions in Lemma \ref{lem: lower adpt cond} imply that ``the more disadvantaged group should get
treated.'' Such RDD settings can be easily found in the education
literature where students or schools with lower academic performance receive
some kind of support (\citealp{chay2005central}; \citealp{chiang2009accountability}; \citealp{jacob2004remedial}; \citealp{leuven2007effect}; \citealp{matsudaira2008mandatory}), in the environmental economics literature where
counties with high pollution levels are exposed to some environmental
regulations (\citealp{chay2005does}; \citealp{greenstone2008does}), and in poverty programs which provide funds to those in need (\citealp{ludwig2007does}). Lastly, we note that when the mean potential
outcome functions are decreasing rather than increasing, we may simply switch
the conditions for $\mathcal{X}_t$ and $\mathcal{X}_c$ in Lemma \ref{lem: lower adpt cond}.
\subsection{Class of adaptive procedures}
\label{subsec: one-sided}

Under treatment allocation rules considered in the previous section, it is possible to construct a one-sided CI that is optimal for a single function space $\mathcal{F}(C')$ for some $C'$, with its length independent or nearly independent of the larger Lipschitz constant $C$. Ideally, we would use a one-sided CI that performs well over a range of function spaces, corresponding to a range of Lipschitz constants $C' \in [\underline{C}, \overline{C}]$ for some specified bounds $\underline{C}$ and $\overline{C}$. To this end, we define a class of adaptive procedures so that our procedure is
based on multiple CIs which solve the optimization problem (\ref{eq: one-sided
  mm problem}) for different values of Lipschitz constants $C' \in \{C_j\}_{j = 1}^J$, with
$\underline{C} \leq C_1 < \cdots < C_J \leq \overline{C}$. For now, we assume
$(C_j)_{j = 1}^J$ are given. We discuss how to choose this sequence of Lipschitz constants
in a way that is optimal later on.

To introduce our procedure, we first characterize the solution to (\ref{eq:
  one-sided mm problem}), applying the general result of
\cite{armstrong2018optimal} to our setting. Given a value of $C'$, it turns out
that the lower CI which solves (\ref{eq: one-sided mm problem}) is based on a
linear estimator $\widehat{L}(C') := a(C') + \sum_{i = 1}^n w_i(C') y_i$ for
non-negative weights $w_i(C')$ and a centering constant $a(C')$. Given some
$\tau \in (0, 1)$, the endpoint of the lower CI which maintains the uniform
coverage probability $1 - \tau$ over $\mathcal{F}(C')$, takes the form of
\begin{equation}\label{eq: one_sided_CI}
  \hat{c}^L_\tau (C') = \widehat{L}_\tau(C') - \underset{ f \in \mathcal{F}
    \left(C \right) } {\sup}\text{bias}_f ( \widehat{L}_\tau(C') ) -
  z_{1-\tau}\,\text{sd} (\widehat{L}_\tau(C') ).    
\end{equation}
Here, we make the dependence on $\tau$ explicit because we later choose a
suitable $\tau < \alpha$ when $J > 1$.

For the estimator, we take $\widehat{L}_\tau(C')$ to be the difference between
two kernel regression estimators, say
$\widehat{L}_{t, \tau}(C') - \widehat{L}_{c, \tau}(C')$, where
\begin{align*}
  \widehat{L}_{q, \tau}(C') := \frac{\sum_{x_i \in \mathcal{X}_q} K(x_i, h_{q, \tau}(C')) y_i / \sigma^2(x_i)}{\sum_{x_i \in \mathcal{X}_q} K(x_i, h_{q, \tau}(C')) / \sigma^2(x_i)}, \text{ for each } q \in \{t, c\}.
\end{align*}
The optimal bandwidths are given by
\begin{align}
  \label{eq: h_q,tau implementation}
  \begin{aligned}
    h_{t, \tau}(C') &= \omega_{t}^* \left( z_{1 - \tau}; C, C' \right) \cdot
    \left( {1}/{C}, {1}/{C'}
    \right) \text{ and} \\
    h_{c, \tau}(C') &= \omega_{c}^* \left( z_{1 - \tau}; C', C \right) \cdot
    \left( {1}/{C'}, {1}/{C} \right),
  \end{aligned}
\end{align}
which completes the definition of the estimator $\widehat{L}_\tau(C')$. The
worst-case bias and the standard deviation of $\widehat{L}_\tau(C')$ are given
by
\begin{align}      \label{eq: sd bias adpt}
  \begin{aligned}
    \text{sd} ( \widehat{L}_\tau(C') ) &= s(z_{1 - \tau}, h_{t, \tau}(C'); C,
    C', \sigma^2(\cdot))
    \\
    \sup_{f \in \mathcal{F}(C)} \text{bias}( \widehat{L}_\tau(C') ) &= a_{t}
    (h_{t, \tau}(C'); C, C') - a_{c} (h_{c, \tau}(C'), C', C) \\[-1.2ex]
    & + \frac{1}{2} [ \omega_{t}^* \left( z_{1 - \tau}; C, C' \right) +
      \omega_{c}^* \left( z_{1 - \tau}; C', C \right) - z_{1 - \tau} \,\text{sd} (
      \widehat{L}_\tau(C') ) ],
  \end{aligned}
\end{align}
which concludes the form of the one-sided CI which solves (\ref{eq: one-sided mm
  problem}).

The CI $\left[\hat{c}_\tau^L(C'), \infty \right)$ is in terms of the worst-case
excess length over $\mathcal{F}(C')$ for a single Lipschitz constant
$C'$. We construct a CI which performs well over a
collection of function spaces by
taking intersection of CIs in the form of
$\left[\hat{c}_\tau^L(C'), \infty \right)$ for different values of $C'$. Note
that taking the intersection ``picks out'' the shortest CI among multiple CIs
formed. This is roughly equivalent to inferring from the data the true function
space where the regression function belongs to, and using the CI which performs
well over that function space. 

The value of $\tau$ should be calibrated so that the resulting CI maintains
correct coverage probability $1 - \alpha$ after taking the intersection. Suppose
we have a collection of $J$ CIs
$\left\{ \left[\hat{c}_\tau^L(C_j), \infty \right) \right\}_{j = 1}^J$. A simple
procedure to take the intersection of such CIs is to use a Bonferroni procedure
and take $\tau = \alpha / J$. This procedure, however, is conservative since the
correlations among the estimators
$\widehat{L}_\tau(C_1),...,\widehat{L}_\tau(C_J)$ are positive, and highly so
when $C_{j}$ and $C_{k}$ are close. To calibrate the value of $\tau$ taking into
account such positive correlation, let $(V_{j,\tau})_{j = 1}^J$ denote a
$J$-dimensional multivariate normal random variable centered at zero, unit
variance, and with covariance terms given as
\begin{align}
  \label{eq: Cov V_j tau}
  \begin{aligned}
    \text{Cov}(V_{j,\tau}, V_{k,\tau}) = &
    \frac{\sum_{x_i \in \mathcal{X}_t} K(x_i, h_{t, \tau}(C_j)) K(x_i, h_{t, \tau}(C_k)) / \sigma^2(x_i)}{z_{1 - \tau}^2 / \omega_{t}^* \left( z_{1 - \tau}; C, C_j \right) \omega_{t}^* \left( z_{1 - \tau}; C, C_k \right)}  \\
    &+ \frac{\sum_{x_i \in \mathcal{X}_c} K(x_i, h_{c, \tau}(C_j)) K(x_i, h_{c,
        \tau}(C_k)) / \sigma^2(x_i)}{z_{1 - \tau}^2 / \omega_{c}^* \left( z_{1 -
          \tau}; C_j, C \right) \omega_{c}^* \left( z_{1 - \tau}; C_k, C
      \right)}.
  \end{aligned}
\end{align}

Then, we can show that if we take $\tau^*$ to be the value of $\tau$ that solves
\begin{align}
  \label{eq: tau* def}
  P \left( \max_{1 \leq j \leq J} V_{j, \tau} > z_{1 - \tau}, 
  \right) = \alpha,
\end{align}
then the resulting CI obtained by taking the intersection has correct
coverage. Regarding the solution to the above equation, we can actually show
that $\text{Cov}(V_{j,\tau}, V_{k,\tau})$ does not depend on $\tau$ as
$n \to \infty$, which implies that $ \max_{1 \leq j \leq J} V_{j, \tau^*}$
converges in distribution to a random variable $V_{\max}$ which does not
dependent on $\tau$.\footnote{In a simulation exercise not reported, we find
  this asymptotic approximation works well for a moderate sample size such as
  $n = 100$.}  
The following is the main theoretical result for our intersection CI.
\begin{theorem}
  \label{thm: adpt-CI}
  Given $\{C_1,...,C_J\}$, $C$ and $\alpha \in (0,1)$, let $\tau^*$ be defined
  as in (\ref{eq: tau* def}).  Then, under Assumption \ref{assn: Gaussian temp},
  if we let $\hat{c} := \max_{1 \leq j \leq J} \hat{c}_{\tau^*}^L(C_j)$, we have
  \begin{align*}
    \inf_{f \in \mathcal{F}(C)}
    P_f \left(L_{RD} f \in [\hat{c}, \infty) \right)
    \geq 1 - \alpha.
  \end{align*}
\end{theorem}
The next result, which is immediate from \cite{armstrong2018optimal}, shows that
the class of adaptive procedures we consider is a reasonable one. It basically
states that each of the CIs $[\hat{c}_{\tau^*}^L(C_j), \infty)$ is an optimal CI
when $f \in \mathcal{F}(C_j)$, except that it covers the true parameter with
probability $1 - \tau^*$ instead of $1 - \alpha$, which is the price we pay to
adapt to multiple Lipschitz classes.
\begin{cor}
  Under Assumption \ref{assn: Gaussian temp},
  $[\hat{c}_{\tau^*}^L(C_j), \infty)$ solves
  \begin{align*}
    \inf_{[\hat{c}, \infty) \in \mathcal{L}_{\tau^*}(C)} 
    \sup_{f \in \mathcal{F}(C_j)}
    E_f[L_{RD} f - \hat{c}].
  \end{align*}
\end{cor}
Again, considering the simple case where $d = 1$ and $\mathcal{V} = \{1\}$
provides an intuition behind our adaptive procedure. Under this simple case,
$\widehat{L}_{q,\tau}(C')$ is simply a kernel regression estimator with kernel
$K(z) = [1 - |z|]_+$ and bandwidths
\begin{align}
  \label{eq: h_q,tau d = 1}
  h_{t, \tau}(C') = \omega_{t}^* \left( 
  z_{1 - \tau};
  C, C' \right) / C',\quad
  h_{c, \tau}(C') = \omega_{c}^* \left( 
  z_{1 - \tau};
  C', C \right) / C'.
\end{align}
Applying results of \cite{kwonkwon2019regpoint}, we can show that these
quantities are approximately $c \times (C')^{-2/3}$ for some constant $c$ for
large $n$. This implies that we construct estimators with varying bandwidth
sizes, and compare the lengths of resulting one-sided CIs. The estimator with a
smaller bandwidth is the one constructed to perform well when the Lipschitz
constant of the regression function is large, and vice versa for the estimator
with a larger bandwidth. This is because when the Lipschitz constant of the
regression function is large, the excess length of a one-sided CI is reduced by
taking a smaller bandwidth to decrease the absolute size of the bias. On the
other hand, when the Lipschitz constant of the regression function is small,
reducing the standard deviation by taking a larger bandwidth matters more in
reducing the excess length of a one-sided CI. This idea is similar to the
bandwidth snooping procedure suggested by \cite{armstrong2018simple}. The results
therein focus on the case with a single running variable and adapting to the
H\"older exponent.
\begin{remark} [Specifying $C$]
  \label{rem: adpt C Inf}
  When $d = 1$, the worst-case bias is 0 under the treatment allocation rules in Section \ref{subsec: direction}. Therefore, a researcher can set $C = \infty$ and
  not worry about correct coverage. On the other hand, when $d > 1$, the size of $C$ governs how large weights should be on observations outside $[0, \infty)^d$ and $(-\infty, 0]^d$. A larger $C$ puts smaller weights on those observations, leading to a wider CI, with $C = \infty$ corresponding to using observations only in $[0, \infty)^d$ and $(-\infty, 0]^d$. However, the length of the adaptive CI shrinks when the true regression function has a smaller first derivative bound even if $C$ is set to be a large number, alleviating the concern that large $C$ might lead to a less informative CI. 
\end{remark}
\begin{remark} [Adaptation in multi-dimensional RDDs]
  Consider the setting of Remark \ref{rem: multi-dim Lip}, where the Lipschitz
  continuity is specified by the weighted $\ell_1$-norm with Lipschitz constants
  $C_1,...,C_d$. Then, our adaptive procedure allows one to consider different
  values of $C$ other than $C = 1$. Note that this takes $C_1 / C_s$ given for
  $s = 2,...,d$. Adapting to different values of $C_1 / C_s$ is an interesting
  extension that we did not pursue here.
\end{remark}
\subsection{Choice of an adaptive procedure}

In this section, we discuss the choice of Lipschitz constants $(C_j)_{j = 1}^J$
to conclude our definition of the adaptive one-sided CI. The choice of function
spaces to adapt to is especially relevant in our setting unlike the previous
literature on adaptive inference. The previous literature has mostly focused on
rate-adaptation, the problem of how to construct a CI which shrinks at an
optimal rate. For example, \cite{armstrong2015adaptive} and
\cite{armstrong2018simple} discuss adaptive testing and construction of adaptive
CIs for the RD parameter, which adapt to H\"older exponents $\beta \in (0,
1]$. In this case, since different H\"older exponents imply different
convergence rates, adapting to the entire continuum of H\"older exponents is in
some sense always optimal. On the other hand, when we fix $\beta = 1$ and try to
adapt to Lipschitz constants, the convergence rate is always $n^{1/(2+d)}$, and
what matters is the actual length of CIs, not their rate of convergence.

The optimal but infeasible adaptive CI is given as
$CI = [\hat{c}, \infty)$ such that
\begin{align*}
  \sup_{f \in \mathcal{F}(C')} E[L_{RD} f - \hat{c}] = \ell(C'; C),
\end{align*}
for all $C' \in [\underline{C}, \overline{C}]$, where $\ell(C'; C)$ is defined
in (\ref{eq: one-sided mm problem}). This is infeasible since the form of CI
which is optimal over $\mathcal{F}(C')$ is different from the one which is
optimal over $\mathcal{F}(C'')$ for $C' \neq C''$. Instead, our aim is to
construct a CI $[\hat{c}, \infty)$ such that
\begin{align*}
  \ell^{\text{adpt}}(C') := \sup_{f \in \mathcal{F}(C')} E[L_{RD} f - \hat{c}]
\end{align*}
is close to $\ell(C'; C)$ over
$C' \in \left[\underline{C}, \overline{C} \right]$ given some
$\left( \underline{C}, \overline{C} \right)$, when $\ell^{\text{adpt}}(C')$ and
$\ell(C'; C)$ are viewed as a function of $C'$.

\begin{table}[t!]
  \begin{tabularx}{\textwidth}{X}
    \hline
    \textbf{Procedure 2} \emph{Adaptive One-sided CI}. \\
    \hline
    \begin{minipage}[t]{\linewidth}
      \begin{enumerate}
      \item Choose $C$ so that the Lipschitz continuity in (\ref{eq:
          Lipschitz}) is satisfied, with a suitable choice of a norm $||\cdot||$
        on $\mathbb{R}^d$ when $d > 1$.\vspace{-5pt} 
        \begin{itemize}
            \item[-] When $d = 1$, we can set $C = \infty$; see Remark \ref{rem: adpt C
            Inf} in Section \ref{subsec: one-sided}.
        \end{itemize}
      \item Choose values of $\underline{C}$ and $\overline{C}$ such that
        $0 \leq \underline{C} < \overline{C} \leq C$, which is the region where the
        adaptive CI will be close to optimal.\vspace{-5pt} 
        \begin{itemize}
            \item[-] A reasonable value of $\underline{C}$ can be
          estimated; see Appendix \ref{sec: minC-est}.
        \end{itemize}
      \item Let $\mathcal{C}(J)$ be the
        equidistant grid of size $J$ over
        $\left[\underline{C}, \overline{C} \right]$. Starting from $J=2$, increase $J$ by $1$ until $\lvert\Delta(\mathcal{C}(J)) - \Delta(\mathcal{C}(J-1))\rvert \leq \varepsilon$, for a tolerance level $\varepsilon$.
      \item Let $J^*$ be the value of $J$ we stopped at in Step 3. We use $\mathcal{C}(J^*)= (C_j)_{j=1}^{J^\ast}$ as the sequence of Lipschitz
        constants to construct the adaptive CI.
        \item Using \eqref{eq: one_sided_CI}, \eqref{eq: h_q,tau implementation}, \eqref{eq: sd bias adpt}, and \eqref{eq: tau* def}, finally obtain the adaptive one-sided CI,
$$\Big[\max_{1 \leq j \leq J} \hat{c}_{\tau^*}^L(C_j), \infty \Big).$$
      \end{enumerate}
    \end{minipage} \\
    \\
    \hline
  \end{tabularx}
\end{table}

By restricting the class of one-sided CIs to those considered in Section
\ref{subsec: one-sided}, we can show that there is a measure of distance between $\ell^{\text{adpt}}(\cdot)$ and
$\ell(\cdot; C)$ which is both reasonable and easy to calculate. To be specific, let $\mathcal{C} = (C_j)_{j = 1}^J$ denote some sequence of Lipschitz
constants to be used to construct our adaptive CI, and denote the endpoint of
the CI with $\hat{c}(\mathcal{C})$. Then, write
\begin{align*}
  \ell^{\text{adpt}}(C'; \mathcal{C}) := \sup_{f \in \mathcal{F}(C')} E[L_{RD} f - \hat{c}(\mathcal{C})].
\end{align*}
Consider the following quantity measuring the ``distance'' between $\ell^{\text{adpt}}(\cdot; \mathcal{C})$ and $\ell(\cdot; C)$:
\begin{align}
  \label{eq: best C_j's}
  \Delta(\mathcal{C}) :=
  \sup_{C' \in \left[\underline{C}, \overline{C} \right]}
  \frac{\ell^{\text{adpt}}(C'; \mathcal{C})}{\ell(C'; C)}.
\end{align}
Note that this criterion is consistent with the previous literature which compares the performances of different confidence intervals using a ratio measure (\citealp{cai2004adaptation}; \citealp{armstrong2018optimal}). Specifically, $\Delta(\mathcal{C})$ satisfies 
\begin{align*}
    \ell^{\text{adpt}}(C'; \mathcal{C}) \leq \Delta(\mathcal{C}) \ell(C'; C),\ \forall C' \in \left[\underline{C}, \overline{C} \right].
\end{align*}
This is the precisely the notion of adaptive CIs introduced in \cite{cai2004adaptation}. Since the rates of $\ell^{\text{adpt}}(C'; \mathcal{C})$ and $\ell(C'; C)$ are the same under our setting, the choice of an adaptive procedure should be based on the size of the constant $\Delta(\mathcal{C})$. 

An advantage from focusing on the class of adaptive CIs proposed in Section \ref{subsec: one-sided} is that it is easy to evaluate $\Delta(\mathcal{C})$
as shown in the following proposition.\\
\begin{prop}
  \label{prop: ell^adpt calc}
  We have
  $\ell^{\text{adpt}}(C'; \mathcal{C}) = E \min_{j\leq J} U_{j}$, where
  $U := (U_1,\dots, U_J)'$ is a Gaussian random vector with known
  mean and variance.\footnote{The expressions for the mean and variance are given in the proof.} 
  Furthermore, we have
  \begin{align*}
    \ell(C';C) = \omega_t^*(z_{1 - \alpha}, C, C') + \omega_c^*(z_{1 - \alpha}, C', C).
  \end{align*}
\end{prop}
Based on the discussion so far, we recommend choosing $\{C_j\}_{j = 1}^J$ based on the value of $\Delta(\mathcal{C})$. When calculating this value, searching for all possible sequences $\mathcal{C}$ is infeasible in practice. Instead, we suggest taking
$\mathcal{C} = \mathcal{C}(J)$ to be the equidistant grid of size $J$ on
$\left[\underline{C}, \overline{C} \right]$ and calculating $\Delta(\mathcal{C}(J))$ for different values of $J$. Our simulation study (not reported) suggests that the gain from increasing $J$ becomes very small after some threshold. Hence, a computationally attractive procedure is to increase the value of $J$ until the additional gain from using $J + 1$ instead of $J$ is smaller than a tolerance parameter. In the simulation designs of Section \ref{sec: Monte Carlo}, the largest value of $\Delta(\mathcal{C}(J^*))$ across all scenarios is less than 1.07, where $J^*$ is the value of $J$ obtained by the method described above. That is, the worst-case length of the adaptive CI is at most $7$\% longer than the worst-case length of the CI that we would have used if we knew the true Lipschitz constant, in the data generating processes considered in Section \ref{sec: Monte Carlo}.

Note that our procedure
requires a researcher to specify the values of $\underline{C}$ and
$\overline{C}$. For $\underline{C}$, while it can be set to 0, a reasonable
value of $\underline{C}$ can be estimated from the data as discussed in Appendix
\ref{sec: minC-est}. The value of $\overline{C}$ defines a region for Lipschitz
constants where the adaptive CI is intended to perform well. Therefore, a
researcher may choose $\overline{C}$ to be some non-conservative potential value
of the Lipschitz constant of the regression function. As discussed above, the
coverage probability is not affected whatever the value of $\overline{C}$ is, so
it will not be a burdensome task for a researcher to choose
$\overline{C}$. Procedure 2 summarizes the steps for constructing the adaptive one-sided CI, including the choice of the Lipschitz spaces to adapt to.
\section{Monte Carlo Simulations}
\label{sec: Monte Carlo}

\begin{table}[b]
  \centering
  \begin{tabular}{c|ccc}
    \hline
    & $x_i \sim \text{unif}(-1,1)$ & Constant variance & $C$ is small \\ 
    \hline
    Design 1 & YES & YES & YES \\ 
    Design 2 & NO & YES & YES \\ 
    Design 3 & YES & NO & YES \\ 
    Design 4 & NO & NO & YES \\ 
    Design 5 & YES & YES & NO \\ 
    Design 6 & NO & YES & NO \\ 
    Design 7 & YES & NO & NO \\ 
    Design 8 & NO & NO & NO \\ 
    \hline
  \end{tabular}
  \caption{Simulation design specifications}
  \label{tab: designs}
\end{table}

We investigate the performance of our minimax and adaptive procedures via a
simulation study. We focus on the case where $d = 1$ with individuals being
treated if $x_i < 0$. We restrict the support of $x_i$ to be $[-1, 1]$. For a
given true sharp RDD parameter $\theta \equiv L_{RD} f$, we consider the
following designs. 
\begin{enumerate}
\item Linear design: $f_{1c}(x) = Cx$, $f_{1t}(x) = f_{1c}(x) + \theta$. We have
  $|f_{1q}'(x)| = C$, and $|f_{1q}''(x)| = 0$ for all $x \in [-1, 1]$ and
  $q \in \{t, c\}$.
\item Modified specification of \cite{armstrong2018optimal_supp}: given some
  ``knots'' $(b_1, b_2)$ such that $0 < b_1 < b_2$, define
  \begin{align*}
    f_{2c}(x) = \frac{3C}{2} (x^2 - 2(x - b_1)_+^2 + 2(x - b_2)_+^2), \hspace{15pt}
    f_{2t}(x) = -f_{2c}(-x) + \theta.
  \end{align*}
  If $b_1 \geq b_2/2$, both functions are increasing. Taking
  $(b_1, b_2) = (1/3, 2/3)$ gives $|f_{2q}'(x)| \leq C$, and
  $|f_{2q}''(x)| = 3C$ for all $x \in [-1, 1]$ and $q \in \{t, c\}$. We also
  have $|f_{2q}'(0)| = 0$.
\item Modified specification of \cite{babii2019isotonic}: define
  \begin{align*}
    f_{3c}(x) = \frac{C}{4}(x^3 + x),
    \hspace{15pt}
    f_{3t}(x) = f_{3c}(x) + \theta.
  \end{align*}
  We have $|f_{3q}'(x)| \leq C$ and $|f_{3q}''(x)| = 3C/2$ for all
  $x \in [-1, 1]$ and $q \in \{t, c\}$. We also have $|f_{3q}'(0)| = C/4$ and
  $|f_{3q}''(0)| = 0$.
\item Nonzero first and second derivatives at $0$: define
  \begin{align*}
    f_{4c}(x) = C((3x + 1)^{1/3} - 1),
    \hspace{15pt}
    f_{4t}(x) = -f_{4c}(-x) + \theta.
  \end{align*}
  We have $|f_{4q}'(x)| \leq C$, and $|f_{4q}''(x)| \leq 2C$ for all
  $x \in [-1, 1]$ and $q \in \{t, c\}$. We also have $|f_{4q}'(0)| = C$ and
  $|f_{4q}''(0)| = 2C$.
\end{enumerate}

For the running variables, we consider $x_i \sim \text{unif}(-1, 1)$ and
$x_i \sim 2 \times \text{Beta}(2,2) - 1$. The latter is used by
\cite{babii2019isotonic}, and gives more observations around the
cutoff. Finally, we consider both homoskedastic and heteroskedastic designs,
$\sigma_1(x) = 1$ and $\sigma_2(x) = \phi(x) / \phi(0)$, where $\phi(x)$ is the
standard normal pdf. The sample size is $n = 500$. Figure
\ref{fig:reg fcns 1dim} in Appendix \ref{sec: aux fig} provides plots for the four regression functions.

We estimate the conditional variance based on local constant kernel regression,
where the initial bandwidth is chosen based on Silverman's rule of thumb. This
is to avoid using a bandwidth selection method based on local linear regression
to ensure that our proposed method works even when the second derivative is very
large. After the estimators are calculated based on the estimated conditional
variance, we construct the CIs following \cite{armstrong2020simple}, who use a
simple way to estimate the variance by
\begin{align*}
  \hat{\sigma}(x_i)^2 =
  \frac{J}{J + 1}
  \left(
  y_i - \frac{1}{J} \sum_{i = 1}^J
  y_{j(i)}
  \right)^2,
\end{align*}
for some fixed $J$, where $j(i)$ denotes the index for the closest observation
to $i$ (with the same treatment status). The default value in their
implementation is $J = 3$, which we follow.
\begin{table}[t]
  \centering
  \begin{tabular}{l|ccc|ccc}
    \hline
    & \multicolumn{3}{c|}{$f = f_1$} & \multicolumn{3}{c}{$f = f_2$} \\
    \hline
    & Length:& RBC& Minimax& Length: & RBC& Minimax \\ 
    & RBC/MM & coverage & coverage & RBC/MM & coverage & coverage \\
    \hline
    Design 1 & 1.094 & 0.925 & 0.968 & 1.097 & 0.924 & 0.942 \\ 
    Design 2 & 1.112 & 0.938 & 0.979 & 1.115 & 0.941 & 0.949 \\ 
    Design 3 & 1.082 & 0.926 & 0.968 & 1.085 & 0.924 & 0.942 \\ 
    Design 4 & 1.099 & 0.936 & 0.979 & 1.103 & 0.940 & 0.950 \\ 
    Design 5 & 1.128 & 0.925 & 0.919 & 1.148 & 0.930 & 0.945 \\ 
    Design 6 & 1.137 & 0.938 & 0.934 & 1.162 & 0.941 & 0.951 \\ 
    Design 7 & 1.117 & 0.926 & 0.920 & 1.139 & 0.927 & 0.945 \\ 
    Design 8 & 1.125 & 0.936 & 0.935 & 1.153 & 0.941 & 0.952 \\ 
    \hline
  \end{tabular}
  \caption{Comparison between RBC and minimax; $f \in \{f_1, f_2\}$}
  \label{fig: rbc comp f1 f2}
\end{table}

\subsection{Minimax procedure}

First, we investigate the performance of the minimax procedure described in
Section \ref{sec: minimax}. We consider different combinations of specifications
on the running variable, the variance function, and the value of $C$: 1) $x_i$
follows either a uniform or a Beta distribution, 2) the variance function is
given by either $\sigma^2 \times \sigma_2^1(x)$ or
$\sigma^2 \times \sigma_2^2(x)$ (where $\sigma^2 = 1/4$), and 3) the true value
of $C$ is either small ($C = 1$) or large ($C = 3$). For the minimax two-sided
CI, we consider the case where we correctly specify the Lipschitz constant in
all cases, setting $C = 3$. The results were calculated with $1,000$
repetitions. The definition of each design is given in Table \ref{tab: designs}.

\begin{table}[t]
  \centering
  \begin{tabular}{l|ccc|ccc}
    \hline
    & \multicolumn{3}{c|}{$f = f_3$} & \multicolumn{3}{c}{$f = f_4$} \\
    \hline
    & Length:& RBC& Minimax& Length: & RBC& Minimax \\ 
    & RBC/MM & coverage & coverage & RBC/MM & coverage & coverage \\
    \hline
    Design 1 & 1.090 & 0.925 & 0.949 & 1.093 & 0.924 & 0.968 \\ 
    Design 2 & 1.109 & 0.937 & 0.951 & 1.111 & 0.938 & 0.977 \\ 
    Design 3 & 1.078 & 0.925 & 0.949 & 1.081 & 0.924 & 0.968 \\ 
    Design 4 & 1.096 & 0.936 & 0.953 & 1.098 & 0.936 & 0.977 \\ 
    Design 5 & 1.094 & 0.923 & 0.962 & 1.119 & 0.920 & 0.930 \\ 
    Design 6 & 1.112 & 0.937 & 0.971 & 1.132 & 0.935 & 0.947 \\ 
    Design 7 & 1.082 & 0.924 & 0.961 & 1.107 & 0.924 & 0.930 \\ 
    Design 8 & 1.099 & 0.936 & 0.971 & 1.119 & 0.933 & 0.947 \\ 
    \hline
  \end{tabular}
  \caption{Comparison between RBC and minimax; $f \in \{f_3, f_4\}$}
  \label{fig: rbc comp f3 f4}
\end{table}

We measure the performance of the minimax procedure by comparing it to one
of the robust bias correction (RBC) procedures by
\cite{calonico2015rdrobust}. We report ratios of the average length
of the CI between the RBC and the minimax procedure, as well as their respective
coverage probabilities. For the robust bias correction procedure, we use the
default implementation provided the \texttt{R} package \texttt{rdrobust}.

\begin{figure}[t]
  \centering \includegraphics[width = 0.8\textwidth]{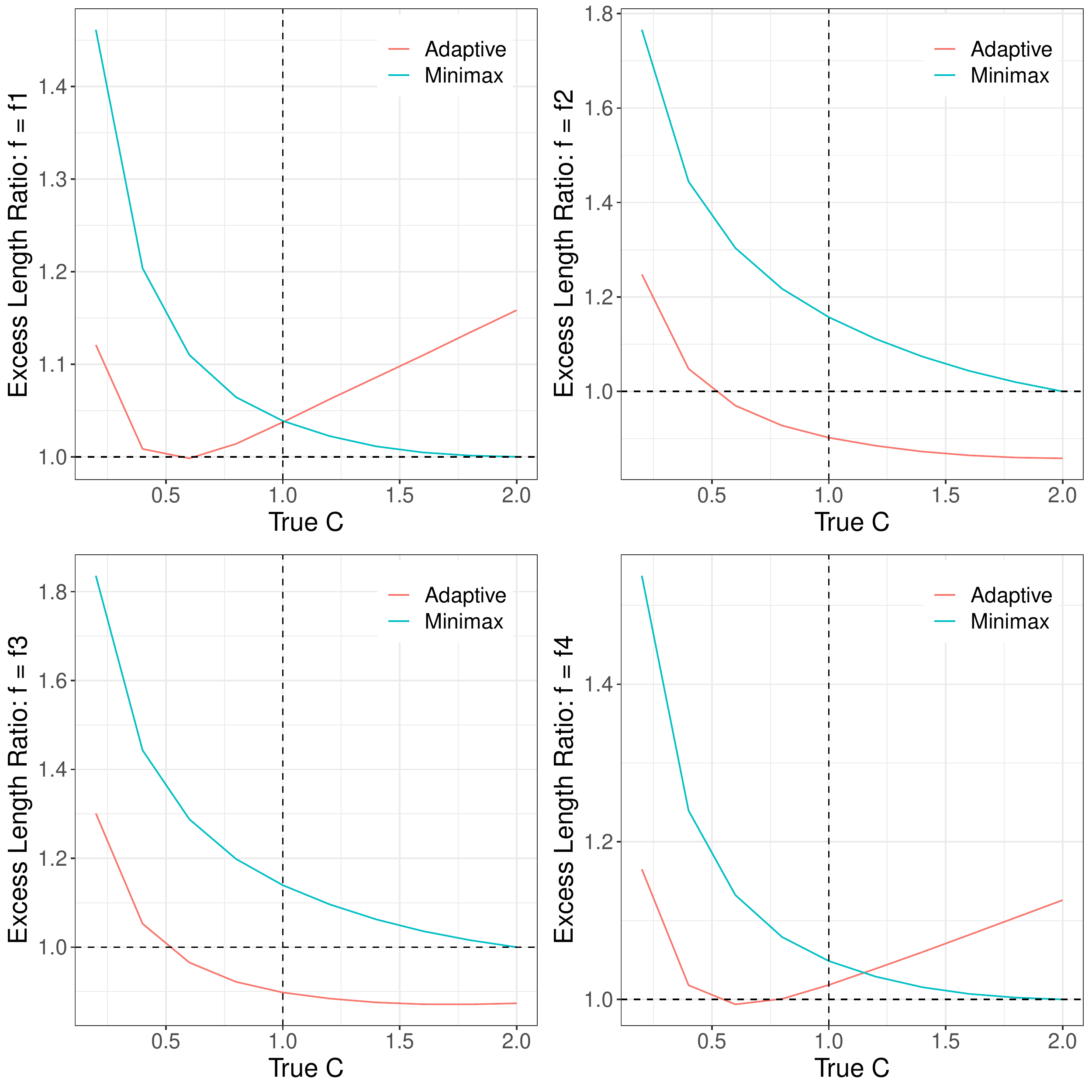}
  \caption{Performance comparison of one-sided CIs}
  \label{fig: adpt comp}
\end{figure}

Tables \ref{fig: rbc comp f1 f2} and \ref{fig: rbc comp f3 f4} display the
results. For each regression function, the first column shows the ratio of the
length of the CIs, and the second and the third columns show the coverage
probabilities of the RBC and the minimax procedures. Despite the fact that the
minimax CIs are shorter than the RBC CIs, the coverage probabilities of the
minimax CIs are closer to the nominal level than the RBC CIs. However, the
length comparison here should not be interpreted as the superiority of our
procedure over the RBC method, since the lengths are sensitive to the choice of
$C$ and to the form of the true regression function. Rather, as we have
discussed so far, the relative advantage of our procedure is the use of a more
transparent function space, and incorporating the monotonicity of the regression
function.

\subsection{Adaptive procedure}

Next, we compare the performance of the one-sided CI to two
benchmarks, the minimax and the oracle one-sided CIs. Specifically, we consider
regression functions $f_1,...,f_4$ and vary the value of $C \in [C_l, C_u]$
which determines their smoothness. For each $C \in [C_l, C_u]$, the oracle
procedure adapts to a single Lipschitz constant $C$, as if we know the true
Lipschitz constant, which is infeasible in practice. On the other hand, the
minimax procedure adapts to the largest Lipschitz constant $C_u$; this procedure
is nearly optimal (among feasible procedures) when we do not have monotonicity,
as shown in \cite{armstrong2018optimal}. Both the oracle and the minimax procedures take monotonicity into account.

We take $(C_l, C_u) = (1/5, 2)$ and consider only Design 1. For the adaptive
procedure, we take $(\underline{C}, \overline{C}) = (1/5, 1)$. We expect that
the performance of the adaptive procedure will be better over $C \in [1/5, 1]$
than over $C \in [1, 2]$.

Figure \ref{fig: adpt comp} shows the relative excess lengths of the minimax and
adaptive CIs compared to the oracle procedure. First, we note that the adaptive
CIs are at most 30\% longer than the oracle procedure, while the minimax CIs are
sometimes as much as 83\% longer than the oracle. Second, for the case of
$f = f_2$ and $f = f_3$, the adaptive CIs are sometimes even shorter than the
oracle. This is because the derivatives of such regression functions near
$x = 0$ are smaller than $C$, although their maximum derivatives are $C$
globally. This shows that the adaptive CI adapts to the \emph{local} smoothness
of the regression function, which is another advantage of using the adaptive
procedure.
\section{Empirical Illustration}
\label{sec: emp}
In this section, we revisit the analysis by \cite{lee2008randomized}. The running variable $x_i \in [-100, 100]$ is the
Democratic margin of victory in each election, and the outcome variable
$y_i \in [0, 100]$ is the Democratic vote share in the next election. The treatment is the incumbency of the Democratic party, with the cutoff given
by 0. Therefore, a positive treatment effect indicates that there is an electoral advantage
to incumbent candidates.

\begin{figure}[th!]
  \centering \includegraphics[width = 0.6\textwidth]{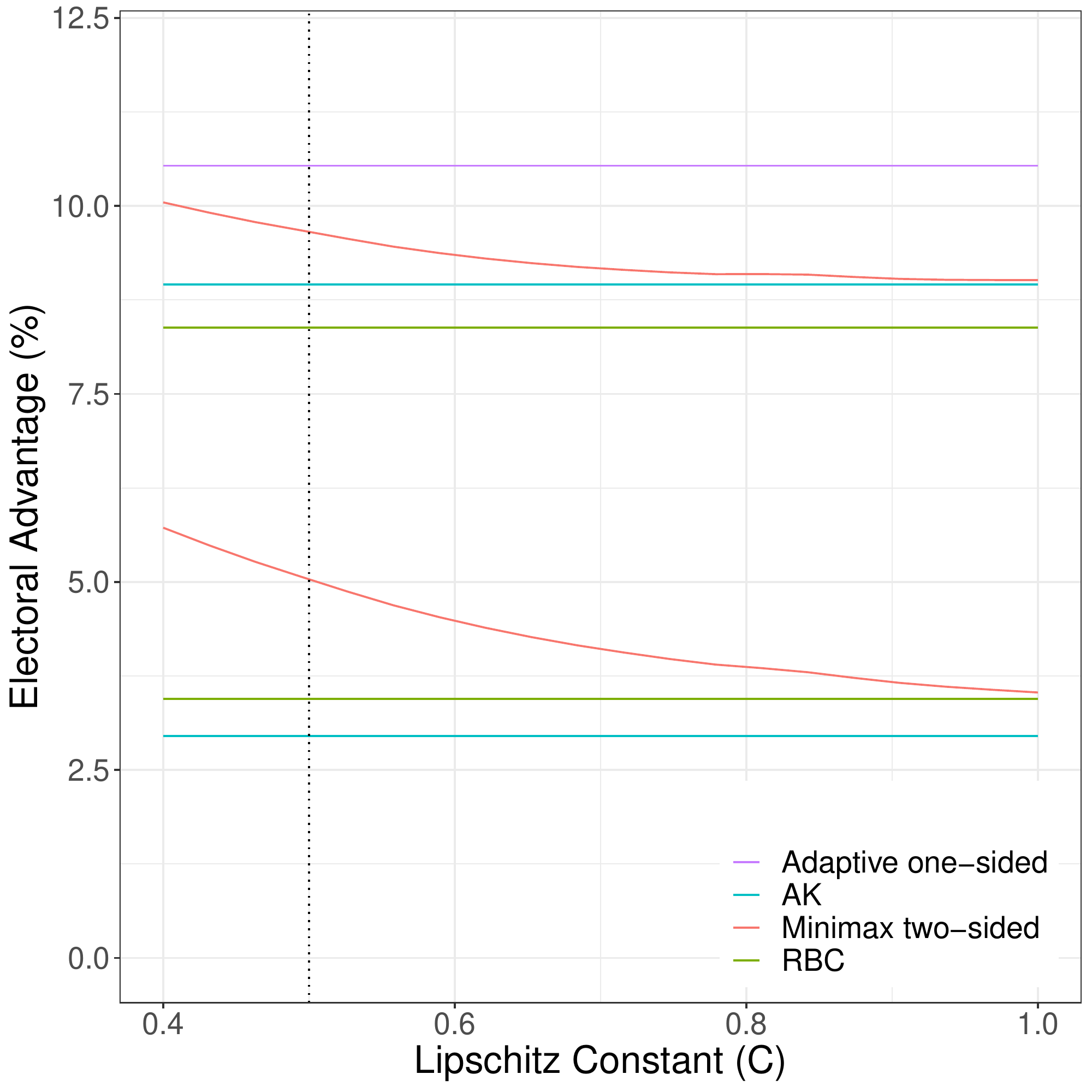}
  \caption{\cite{lee2008randomized} example. The red line (Minimax two-sided) plots our minimax
    optimal CI, while the blue line (RBC) and the green line (AK) plot the CIs
    constructed using the methods of \cite{calonico2015rdrobust} and
    \cite{armstrong2018optimal}, respectively. The purple line (Adaptive one-sided) plots the adaptive one-sided upper CI with $C = \infty$. The vertical dotted line indicates
    $C = 1/2$, our preferred specification.}
  \label{fig: lee 08}
\end{figure}

This empirical application nicely fits the setting we consider in this
paper. For monotonicity, it is plausible to assume that a party's vote share
increases on average in that party's previous vote share. Moreover, for
Lipschitz continuity, it is plausible that a unit increase in the previous
election vote share can predict the vote share increase in the next election by
no more than a single unit. This reflects the viewpoint that the previous
election vote share is a noisy measure of a party's popularity in the current
election. Since $x_i$ is the vote margin instead of the vote share, this
translates to setting $C = 1/2$.

Figure \ref{fig: lee 08} plots the CIs constructed from our minimax optimal
procedure for different values of $C \in [0.4, 1]$, which includes our preferred
specification $C = 1/2$. We also include CIs using the robust bias correction
method of \cite{calonico2015rdrobust} and the minimax CI using a second
derivative bound as in \cite{armstrong2018optimal}. For the latter, the bound on
second derivative is set as $M = 1/10$, which is the largest bound used in their
empirical analysis. The nominal coverage probability is $0.95$ for all CIs. The CIs obtained by the robust bias correction and the second derivative bound are given by $[3.44, 8.38]$ and $[2.95, 8.96]$,
respectively. Our CI has a larger lower bound and upper bound compared to other
methods. Especially, at our preferred specification $C = 1/2$, the CI is given
by $[5.03, 9.66]$, implying a potentially larger size of incumbents' electoral
advantage compared to other inference methods. At this specification, our CI is also relatively shorter than others. For the same dataset, \cite{babii2019isotonic} give a CI of $[6.6, 26.5]$, whose upper bound is larger than other procedures.

As discussed above, since the Lipschitz constant C  has a clear interpretation,
a sensitivity analysis varying the
values of the Lipschitz constant is especially useful in strengthening the
credibility of the inference results. Specifically, we can calculate the size of
the smoothness parameter under which the effect of incumbency becomes
insignificant. It turns out that our minimax CI contains $0$ when $C$ is
larger than 16. The first derivative bound $C = 16$ roughly implies that one
unit increase in the previous election vote share can predict as much as a 32 unit
increase in the next election vote share, which is quite a large
amount. Therefore, we conclude that the significance of the incumbency effect is
robust over a reasonable set of assumptions on the unknown regression function.

Given that various CIs considered here are obtained from different assumptions on the regression function, it would be an interesting exercise to see what we can infer about the RD parameter under a minimal assumption. To this end, we construct an adaptive one-sided CI which maintains coverage over all monotone functions. The discussion in Section \ref{subsec: direction} implies that we can construct an adaptive upper CI in this empirical setting. We take $(\underline{C}, \overline{C}) = (0.1, 0.5)$, reflecting the belief that the true Lipschitz constant is likely to be less than $0.5$. To make the one-sided CI comparable with the other two-sided CIs, the nominal coverage probability of the upper CI is set as $0.975$, and the resulting upper bound is 10.52\%. While this upper CI maintains coverage over all monotone functions, the resulting upper bound is not significantly greater than the upper bounds of the other CIs, which reflects the adaptive nature of the upper CI.

\section{Conclusion}
\label{sec: conc}
In this paper, we proposed a minimax two-sided CI and an adaptive
one-sided CI when the regression function is assumed to be monotone and has
bounded first derivative. We showed our procedure achieves uniform coverage
under easy-to-interpret conditions and can be used to construct either a
two-sided CI with a minimax optimal length or a one-sided one whose excess
length adapts to the smoothness of the unknown regression function. There are
two extensions that we find interesting.\\[-2ex]

\noindent \textbf{Fuzzy RDDs.} There are various RDD applications where
compliance to a treatment status is only partial. Therefore, it would be of
interest to extend our approach to the fuzzy RDD setting, for example, by making
monotonicity and Lipschitz continuity assumptions on the treatment propensity
$p(x) = P(t_i = 1|x_i = x)$, where $t_i$ is the treatment indicator for
individual $i$. This approach will complement the minimax optimal approaches to
fuzzy RDDs using second derivative bounds by \cite{armstrong2020simple} and
\cite{noack2019bias}.\\[-2ex]

\noindent \textbf{Weigthed CATE. } In multi-score RDDs, \cite{imbens2019optimized} suggest estimating the weighted average of conditional treatment effects over different boundary points to make inference more precise. Since the weighted average parameter is a linear functional of the regression function, we can also adjust our framework to conduct inference on the parameter. On the other hand, a closed-form solution might not exist, and how to computationally construct the confidence interval for the weighted average parameter under our setting seems to be an interesting research question.
\newpage
\bibliographystyle{ecta} \bibliography{KK_RDD}
\newpage{}
\appendix
\section{Lemmas and Proofs}
\label{app: lem and proof}
In this section, we collect auxiliary lemmas and omitted proofs. Before
presenting the results, we state the following definition.
\begin{definition}
  \label{def: mod RD}
  Given two Lipschitz constants $C_{1}$ and $C_{2}$, and some positive constant
  $\delta \geq 0$, we define
  \begin{eqnarray}
    \omega \left( \delta; C_{1}, C_{2} \right) 
    & := & 
           \underset{f_{1} \in  \mathcal{F} (C_1),\ f_{2} \in \mathcal{F} (C_2)} 
           {\sup} L_{RD} f_{2} - L_{RD} f_{1} \nonumber \\
    & \text{ } & \text{s.t.}
                 \sum_{i=1}^{n} \left( 
                 \frac{f_{2} (x_{i}) - f_{1} (x_{i})}{\sigma (x_{i})}
                 \right)^{2} \leq \delta^{2}. \label{eq:mod_of_cont}
  \end{eqnarray}
  The quantity $\omega \left( \delta; C_{1}, C_{2} \right)$ is called the
  \emph{ordered modulus of continuity} of $\mathcal{F} (C_1)$ and
  $\mathcal{F} (C_2)$ for the parameter $L_{RD}f$.
\end{definition}
\subsection{Lemmas}

%
%
\begin{lemma}
  \label{lem: Inv_mod}
  Given some pair of numbers $C, C' \geq 0$ and $\delta \geq 0$, define
  $\omega_q(\delta; C, C')$ for each $q \in \{t,c\}$ to be the solution to the
  following equation
  \begin{align*}
    \sum_{x_{i} \in \mathcal{X}_{q}} \left[ \omega_q(\delta; C, C') - C \left\lVert (x_{i})_{\mathcal{V}+} \right\rVert - C'  \left\lVert (x_{i})_{\mathcal{V}-} \right\rVert \right]_{+}^{2} / \sigma^{2}(x_{i}) 
    = \delta^{2}.
  \end{align*}
  Then, we have
  \begin{align}
    \omega_q(\delta; C, C') = \underset{
    f_{1} \in \mathcal{F} (C),
    f_{2} \in \mathcal{F} (C')}
    {\sup} f_{2} (0) - f_{1} (0)\nonumber
    \\ \text{s.t. }
    \sum_{x_{i} \in \mathcal{X}_{q}}
    \left( \frac{f_{2} (x_{i}) - f_{1} (x_{i})}{\sigma (x_{i})}
    \right)^{2} \leq \delta^{2}. \label{eq:Mod_def_f(0)}
  \end{align}
\end{lemma}
\begin{proof}
  See \cite{kwonkwon2019regpoint}. 
\end{proof}
\begin{lemma}
  \label{lem:Mod_RD_form}
  Given some pair $C_1, C_2 \geq 0$ and $\delta \geq 0$, we have
  \[
    \omega \left( \delta; C_{1}, C_{2} \right) = \omega_{t} \left(
      \delta_{t}^{*} \left( \delta; C_{1}, C_{2} \right); C_{1}, C_{2} \right)+
    \omega_{c} \left( \delta_{c}^{*} \left( \delta; C_{2}, C_{1} \right); C_{2},
      C_{1} \right).
  \]
\end{lemma}
\begin{proof}[Proof of Lemma \ref{lem:Mod_RD_form}]
  Noting that
  \begin{eqnarray*}
    L_{RD} f_{2} - L_{RD} f_{1} & = & \left( f_{2, t}(0) - f_{2, c}(0) \right)-\left( f_{1, t}(0) - f_{1, c}(0) \right) \\
                                & = & \left( f_{2, t}(0) - f_{1, t}(0) \right) + \left( f_{1, c}(0) - f_{2, c}(0) \right),
  \end{eqnarray*}
  $\omega \left( \delta; C_{1}, C_{2} \right)$ is obtained by solving the
  following problem:
  \begin{align*}
    & \sup_{f_{1, t}, f_{1, c}, f_{2, t}, f_{2, c}}
      \left( f_{2, t}(0) - f_{1, t}(0) \right)
      +\left( f_{1, c}(0) - f_{2, c}(0) \right) \\
    \text{s.t. } & \sum_{i=1}^{n} 
                   \left( \mathbbm{1} \left\{ x_{i} \in \mathcal{X}_{t} \right\} \left( 
                   \frac{f_{2, t}(x_{i}) - f_{1, t}(x_{i})}{\sigma (x_{i})} \right)^{2}
                   + \mathbbm{1} \left\{ x_{i} \in \mathcal{X}_{c} \right\} 
                   \left( 
                   \frac{f_{1, c}(x_{i}) - f_{2, c}(x_{i})}{\sigma (x_{i})} \right)^{2} 
                   \right) \leq \delta^{2}, \\
    & f_{1, t}, f_{1, c} \in \Lambda_{+, \mathcal{V}} 
      \left( C_{1} \right), 
      f_{2, t}, f_{2, c} \in \Lambda_{+, \mathcal{V}} 
      \left( C_{2} \right).
  \end{align*}
  Using the definitions for $\omega_{t} \left( \delta_{t}; C_{1}, C_{2} \right)$
  and $\omega_{c} \left( \delta_{c}; C_{1}, C_{2} \right)$, and Lemma \ref{lem:
    Inv_mod}, we can write
  \begin{align*}
    \omega \left( \delta; C_{1}, C_{2} \right) 
    = \sup_{\delta_{t} \geq 0, \delta_{c} \geq 0, 
    \delta_{t}^{2} + \delta_{c}^{2} = \delta^{2}}
    \omega_{t} \left( \delta_{t}; C_{1}, C_{2} \right) 
    + \omega_{c} \left( \delta_{c}; C_{2}, C_{1} \right),
  \end{align*}
  which concludes the proof.
\end{proof}
\begin{lemma}
  \label{lem: f2 - f1 _ f2 + f1}
  Given some $\left(C_{1}, C_{2} \right) \in \mathbb{R}_{+}^{2}$ and
  $\delta \geq 0$, write $\delta_t^* := \delta_t^*(\delta; C_1, C_2)$,
  $\delta_c^* := \delta_c^*(\delta; C_2, C_1)$.  Then, we can find
  $\left( f_{\delta, 1}^{*}, f_{\delta, 2}^{*} \right) \in \mathcal{F} \left(
    C_{1} \right) \times \mathcal{F} \left( C_{2} \right)$, which satisfy the
  following three conditions: (i)
  $L_{RD} f_{\delta, 2}^{*} - L_{RD} f_{\delta, 1}^{*} = \omega \left( \delta;
    C_{1}, C_{2} \right)$, (ii)
  $\left\lVert\frac{f_{\delta, 2}^{*}-f_{\delta, 1}^{*}}{\sigma}
  \right\rVert_{2}^{2} = \delta^{2}$, and (iii) when we write
  \begin{eqnarray*}
    f_{\delta, 2}^{*} & = & f_{\delta, 2, t}^{*}\mathbbm{1} \left\{ x \in \mathcal{X}_{t} \right\} +f_{\delta, 2, c}^{*}\mathbbm{1} \left\{ x \in \mathcal{X}_{c} \right\} \\
    f_{\delta, 1}^{*} & = & f_{\delta, 1, t}^{*}\mathbbm{1} \left\{ x \in \mathcal{X}_{t} \right\} +f_{\delta, 1, c}^{*}\mathbbm{1} \left\{ x \in \mathcal{X}_{c} \right\} ,
  \end{eqnarray*}
  $\left( f_{\delta, 1, t}^{*} , f_{\delta, 2, t}^{*} \right) \in \mathcal{F}
  \left( C_{1} \right) \times \mathcal{F} \left( C_{2} \right)$ and
  $\left( f_{\delta, 1, c}^{*} , f_{\delta, 2, c}^{*} \right) \in \mathcal{F}
  \left( C_{1} \right) \times \mathcal{F} \left( C_{2} \right)$ satisfy
  \begin{eqnarray}
    f_{\delta, 2, t}^{*} - f_{\delta, 1, t}^{*} 
    & = & \left[ 
          \omega_{t} \left( \delta_{t}^{*}; C_{1}, C_{2} \right)
          - C_{1} \left\lVert (x)_{\mathcal{V}+} \right\rVert 
          - C_{2} \left\lVert (x)_{\mathcal{V}-} \right\rVert \right]_{+} \label{eq: f2t - f1t} \\
    f_{\delta, 1, c}^{*} - f_{\delta, 2, c}^{*} 
    & = & \left[
          \omega_{c} \left( \delta_{c}^{*}; C_{2}, C_{1} \right)
          - C_{1} \left\lVert (x)_{\mathcal{V}-} \right\rVert 
          - C_{2} \left\lVert (x)_{\mathcal{V}+} \right\rVert
          \right]_{+}, \label{eq :f1c - f2c}
  \end{eqnarray}
  \begin{eqnarray}
    f_{\delta, 1, t}^{*} (0) + f_{\delta, 2, t}^{*} (0) 
    & = & 
          \omega_{t} \left( \delta_{t}^{*}; C_{1}, C_{2} \right) \label{eq:f2t - f1t-0} \\
    f_{\delta, 1, c}^{*} (0) + f_{\delta, 2, c}^{*} (0) 
    & = & 
          \omega_{c} \left( \delta_{c}^{*}; C_{2}, C_{1} \right), \label{eq:f2c-f1c-0}
  \end{eqnarray}
  and
  \begin{align}
    & \left( f_{\delta, 2, t}^{*} - f_{\delta, 1, t}^{*} \right) 
      \left( f_{\delta, 1, t}^{*} + f_{\delta, 2, t}^{*} \right) \nonumber \\
    = & \left( f_{\delta, 2, t}^{*} - f_{\delta, 1, t}^{*} \right) 
        \left( \omega_{t} \left( \delta_{t}^{*}; C_{1}, C_{2} \right) 
        + C_{1} \left\lVert (x)_{\mathcal{V}+} \right\rVert 
        - C_{2} \left\lVert (x)_{\mathcal{V}-} \right\rVert \right) \label{eq:f2t+f1c} \\
    & \left( f_{\delta, 1, c}^{*} - f_{\delta, 2, c}^{*} \right) 
      \left( f_{\delta, 1, c}^{*} + f_{\delta, 2, c}^{*} \right) \nonumber \\
    = & \left( f_{\delta, 1, c }^{*} - f_{\delta, 2, c}^{*} \right)
        \left( \omega_{c} \left( \delta_{c}^{*}; C_{2}, C_{1} \right)
        + C_{2} \left\lVert (x)_{\mathcal{V}+} \right\rVert 
        - C_{1} \left\lVert (x)_{\mathcal{V}-} \right\rVert \right). \label{eq:f2c+f1c}
  \end{align}
\end{lemma}
\begin{proof}
  \cite{kwonkwon2019regpoint} show that if we let
  \begin{eqnarray*}
    f_{\delta, 1, t}^{*}(x) & = & 
                                  \begin{cases}
                                    C_{1} \left\lVert (x)_{\mathcal{V}+}
                                    \right\rVert
                                    & \text{if } C_{1} \leq C_{2} \\
                                    \text{min} \left\{ \omega_{t} \left(
                                        \delta_{t}, C_{1}, C_{2} \right) - C_{2}
                                      \left\lVert (x)_{\mathcal{V}-}
                                      \right\rVert,\ C_{1} \left\lVert
                                        (x)_{\mathcal{V}+} \right\rVert \right\}
                                    & \text{otherwise},
                                  \end{cases} \\
    f_{\delta, 2, t}^{*}(x) & = &
                                  \begin{cases}
                                    \text{max} \left\{ \omega_{t} \left(
                                        \delta_{t}, C_{1}, C_{2} \right) - C_{2}
                                      \left\lVert (x)_{\mathcal{V}-}
                                      \right\rVert,\ C_{1} \left\lVert
                                        (x)_{\mathcal{V}+} \right\rVert
                                    \right\}  & \text{if } C_{1} \leq C_{2} \\
                                    \omega_{t} \left( \delta_{t}, C_{1}, C_{2}
                                    \right) - C_{2} \left\lVert
                                      (x)_{\mathcal{V}-} \right\rVert &
                                    \text{otherwise},
                                  \end{cases}
  \end{eqnarray*}
  we have
  $\left( f_{\delta, 1, t}^{*} , f_{\delta, 2, t}^{*} \right) \in \mathcal{F}
  \left( C_{1} \right) \times \mathcal{F} \left( C_{2} \right)$ and these
  functions solve (\ref{eq:Mod_def_f(0)}) for $q = t$, $C = C_1$, and
  $C' = C_2$. Likewise, if we let
  \begin{eqnarray*}
    f_{\delta, 1, c}^{*}(x) & = & 
                                  \begin{cases}
                                    \text{max} \left\{ \omega_{c} \left(
                                        \delta_{c}, C_{2}, C_{1} \right) - C_{1}
                                      \left\lVert (x)_{\mathcal{V}-}
                                      \right\rVert,\ C_{2} \left\lVert
                                        (x)_{\mathcal{V}+} \right\rVert
                                    \right\}  & \text{if } C_{2} \leq C_{1} \\
                                    \omega_{c} \left( \delta_{c}, C_{2}, C_{1}
                                    \right) - C_{1} \left\lVert
                                      (x)_{\mathcal{V}-} \right\rVert &
                                    \text{otherwise},
                                  \end{cases}\\
    f_{\delta, 2, c}^{*}(x) & = &
                                  \begin{cases}
                                    C_{2} \left\lVert (x)_{\mathcal{V}+}
                                    \right\rVert
                                    & \text{if } C_{2} \leq C_{1} \\
                                    \text{min} \left\{ \omega_{c} \left(
                                        \delta_{c}, C_{2}, C_{1} \right) - C_{1}
                                      \left\lVert (x)_{\mathcal{V}-}
                                      \right\rVert,\ C_{2} \left\lVert
                                        (x)_{\mathcal{V}+} \right\rVert \right\}
                                    & \text{otherwise,}
                                  \end{cases}
  \end{eqnarray*}
  we have
  $\left( f_{\delta, 1, c}^{*} , f_{\delta, 2, c}^{*} \right) \in \mathcal{F}
  \left( C_{1} \right) \times \mathcal{F} \left( C_{2} \right)$ and these
  functions solve (\ref{eq:Mod_def_f(0)}) for $q = c$, $C = C_2$, and
  $C' = C_1$. Then, the equations (\ref{eq: f2t - f1t}) -- (\ref{eq:f2c+f1c}) in
  the statement of this lemma follow from the above formulas and Lemma
  \ref{lem:Mod_RD_form}.
\end{proof}
\begin{lemma}
  \label{lem: om_prime}Given some
  $\left(C_{1}, C_{2} \right) \in \mathbb{R}^{2}$ and $\delta \geq 0$, let
  $\omega' \left( \delta; C_{1}, C_{2} \right) = \frac{\partial}{\partial\delta}
  \omega \left( \delta; C_{1}, C_{2} \right)$ and write
  $\delta_t^* := \delta_t^*(\delta; C_1, C_2)$,
  $\delta_c^* := \delta_c^*(\delta; C_2, C_1)$.  Then, we have
  \begin{eqnarray*}
    \omega' \left( \delta; C_{1}, C_{2} \right) & = &
                                                      \frac{\delta}{
                                                      \sum_{x_{i} \in \mathcal{X}_{t}} 
                                                      \left[ \omega_{t} \left( \delta_{t}^{*}; C_{1}, C_{2} \right) 
                                                      - C_{1} \left\lVert (x)_{\mathcal{V}+} \right\rVert 
                                                      - C_{2} \left\lVert (x)_{\mathcal{V}-} \right\rVert 
                                                      \right]_{+} / \sigma^{2}(x_{i})} \\
                                                & = & 
                                                      \frac{\delta}{
                                                      \sum_{x_{i} \in \mathcal{X}_{c}} 
                                                      \left[ \omega_{c} \left( \delta_{c}^{*}; C_{2}, C_{1} \right) 
                                                      - C_{1} \left\lVert (x)_{\mathcal{V}-} \right\rVert 
                                                      - C_{2} \left\lVert (x)_{\mathcal{V}+} \right\rVert 
                                                      \right]_{+} / \sigma^{2}(x_{i})}.
  \end{eqnarray*}
\end{lemma}
\begin{proof}
  Note that $f \in \mathcal{F} (C)$ implies $f + z \in \mathcal{F} (C)$ for any
  $z \in \mathbb{R}$ and $C \geq 0$. Moreover, letting
  $\iota_{t}(x) := \mathbbm{1} \left\{ x \in \mathcal{X}_{t} \right\} $, we have
  $L_{RD} \left( \iota_{t} \right) = 1$. Then, Lemma B.3 in
  \cite{armstrong2016optimal} implies that
  \begin{eqnarray*}
    \frac{\omega' \left( \delta; C_{1}, C_{2} \right)}{\delta} 
    & = & 
          \frac{1}{
          \sum_{i=1}^{n} \left( f_{\delta, 2}^{*}(x_{i}) 
          - f_{\delta, 1}^{*}(x_{i}) \right) 
          \mathbbm{1} \left\{ x \in \mathcal{X}_{t} \right\} / \sigma^{2}(x_{i})} \\
    & = & 
          \frac{1}{
          \sum_{x_{i} \in \mathcal{X}_{t}} 
          \left( f_{\delta, 2, t}^{*}(x_{i}) 
          - f_{\delta, 1, t}^{*}(x_{i}) \right) / \sigma^{2}(x_{i})},
  \end{eqnarray*}
  where $f_{\delta, j}^{*}(x_{i})$ and $f_{\delta, j,t}^{*}(x_{i})$ for
  $j = 1, 2$ are as defined in Lemma \ref{lem: f2 - f1 _ f2 + f1}.  Then, using
  (\ref{eq: f2t - f1t}) in Lemma \ref{lem: f2 - f1 _ f2 + f1}, we get the first
  equality in the statement of this lemma.  Likewise, if we define
  $\iota_{c}(x) = - \mathbbm{1} \left\{ x \in \mathcal{X}_{c} \right\} $, we
  have $L_{RD} \left( \iota_{c} \right) = 1$, and Lemma B.3 in
  \citet{armstrong2016optimal} implies that
  \begin{eqnarray*}
    \frac{\omega' \left( \delta; C_{1}, C_{2} \right)}{\delta} 
    & = & \frac{1}{
          \sum_{x_{i} \in \mathcal{X}_{c}} 
          \left( f_{\delta, 1, c}^{*}(x_{i}) - f_{\delta, 2, c}^{*}(x_{i}) \right) / \sigma^{2}(x_{i})},
  \end{eqnarray*}
  where $f_{\delta, j,c}^{*}(x_{i})$ for $j = 1, 2$ are as defined in Lemma
  \ref{lem: f2 - f1 _ f2 + f1}. Then, using (\ref{eq :f1c - f2c}) in Lemma
  \ref{lem: f2 - f1 _ f2 + f1}, we get the second equality in the statement of
  this lemma.
\end{proof}
\begin{lemma}
  \label{lem: Kernel_sq_sum}
  Given some pair $(C_2, C_1)$ such that $C_2 \geq C_1 \geq 0$, and for some
  $\delta > 0$, write $\delta_t^* := \delta_t^*(\delta; C_2, C_1)$,
  $\delta_c^* := \delta_c^*(\delta; C_1, C_2)$, and
  \begin{align*}
    h_{t, \delta} = \omega_t^*(\delta, C_2, C_1) \cdot \left(
    \frac{1}{C_2}, \frac{1}{C_1}
    \right),
    \hspace{15pt}
    h_{c, \delta} = \omega_c^*(\delta, C_1, C_2) \cdot \left(
    \frac{1}{C_1}, \frac{1}{C_2}
    \right).
  \end{align*}
  Then, we have
  \begin{align*}
    \sum_{x_{i} \in \mathcal{X}_{t}}
    K(x_i, h_{t,\delta})^2 /  \sigma^{2}(x_{i}) 
    &= (\delta_{t}^{*} / \omega_t^*(\delta, C_2, C_1))^2, \\
    \sum_{x_{i} \in \mathcal{X}_{c}}
    K(x_i, h_{c, \delta})^2 /  \sigma^{2}(x_{i}) 
    &= (\delta_{c}^{*} / \omega_c^*(\delta, C_1, C_2))^2.
  \end{align*}
\end{lemma}
\begin{proof}
  Using notations in Lemma \ref{lem: f2 - f1 _ f2 + f1}, we can write
  \begin{align*}
    \sum_{x_{i} \in \mathcal{X}_{t}}
    K \left(x_{i}, h_{t, \delta} \right)^{2} / \sigma^{2}(x_{i})
    = \frac{1}{(\omega_t^*(\delta, C_2, C_1))^{2}} 
    \sum_{x_{i} \in \mathcal{X}_{t}}
    {\left( f_{\delta, 2, q}^{*}(x_{i}) - f_{\delta, 1, q}^{*}(x_{i}) \right)^{2}} / {\sigma^{2}(x_{i})}.
  \end{align*}
  Now, by definitions of $f_{\delta, 1, t}^{*}$ and $f_{\delta, 2, t}^{*}$, we
  have
  $\sum_{x_{i} \in \mathcal{X}_{t}} \left( f_{\delta, 2, t}^{*}(x_{i}) -
    f_{\delta, 1, t}^{*}(x_{i}) \right)^{2} / \sigma^{2}(x_{i}) =
  \delta_{t}^{*2}$, which gives the desired result for the first equation of
  this lemma. The second equation follows from the analogous reasoning.
\end{proof}
\subsection{Proofs of main results}
\begin{proof} [Proof of Lemma \ref{lem: lower adpt cond}]
  Consider a lower confidence interval
  $[\hat{c}, \infty) \in \mathcal{L}_\alpha(C)$. Then, Theorem 3.1 of
  \cite{armstrong2018optimal} implies that
  \begin{align*}
    \sup_{f \in \mathcal{F}(C')}E[L_{RD}f - \hat{c}] \geq 
    \omega(z_{1 - \alpha}, C, C'),
  \end{align*}
  when Assumption \ref{assn: Gaussian temp} holds. The same theorem also implies
  that there exists $[\hat{c}, \infty) \in \mathcal{L}_\alpha(C)$ such that it
  exactly achieves the lower bound above. Therefore, it remains to analyze the
  conditions where $\omega(z_{1 - \alpha}, C, C')$ is bounded by some $A(C')$.

  Now, Lemma \ref{lem:Mod_RD_form} implies that we can write
  \begin{align*}
    \omega(z_{1 - \alpha}, C, C') = \omega_{t} \left( \delta_{t}^{*}, C, C' \right) + \omega_{c} \left( \delta_{c}^{*}, C', C \right),
  \end{align*}
  where $\left( \delta_{t}^{*},\delta_{c}^{*} \right)$ solve
  \begin{align*}
    \sup_{\delta_{t} \geq 0, \delta_{c} \geq 0, \delta_{t}^{2} + \delta_{c}^{2}=z_{1-\alpha}^{2}} \omega_{t} \left( \delta_{t}, C, C' \right) + \omega_{c} \left( \delta_{c}, C, C' \right).
  \end{align*}
  Due to Lemma \ref{lem: Inv_mod},
  $b_{t} = \omega_{t} \left( \delta_{t}^*, C, C' \right)$ solves
  \begin{eqnarray}
    \label{eq: b_t proof of adpt cond}
    \sum_{x_{i} \in \mathcal{X}_{t}} \left[b_{t}- C \left\lVert(x_{i})_{\mathcal{V}+} \right\rVert - C' \left\lVert(x_{i})_{\mathcal{V}-} \right\rVert \right]_{+}^{2} / \sigma^{2}(x_{i}) = (\delta_{t}^*)^{2}.
  \end{eqnarray}
  We first consider sufficiency. Define
  $\underline{\mathcal{X}}_t := (-\infty,0]^{d} \cap \mathcal{X}_t$. Notice that
  for any given $b\in \mathbb{R}$, we have
  \begin{align*}
    & \sum_{x_{i} \in \mathcal{X}_{t}} \left[b- C \left\lVert(x_{i})_{\mathcal{V}+} \right\rVert - C' \left\lVert(x_{i})_{\mathcal{V}-} \right\rVert \right]_{+}^{2} / \sigma^{2}(x_{i})\\
    \geq & \sum_{x_{i}\in\underline{\mathcal{X}}_t} \left[b- C \left\lVert(x_{i})_{\mathcal{V}+} \right\rVert - C' \left\lVert(x_{i})_{\mathcal{V}-} \right\rVert \right]_{+}^{2} / \sigma^{2}(x_{i}).
  \end{align*}
  Therefore, given some $\delta_t \geq 0$, if $b_{t}' = b_{t}'(\delta_t)$ solves
  \[
    \sum_{x_{i}\in\underline{\mathcal{X}_t}} \left[b_{t}'- C
      \left\lVert(x_{i})_{\mathcal{V}+} \right\rVert - C'
      \left\lVert(x_{i})_{\mathcal{V}-} \right\rVert \right]_{+}^{2} /
    \sigma^{2}(x_{i}) = \delta_{t}^{2},
  \]
  $b_{t}' \geq \omega_{t} \left( \delta_{t}, C, C' \right)$ will hold. This
  implies that
  \begin{align*}
    \omega_{t} \left( \delta_{t}^*, C, C' \right)
    \leq \omega_{t} \left( z_{1 - \alpha}, C, C' \right)
    \leq b_{t}'(z_{1 - \alpha}),
  \end{align*}
  where the first inequality is due to the constraint that
  $\delta_{t}^* \geq 0, \delta_{c}^* \geq 0$ and
  $(\delta_{t}^*)^{2} + (\delta_{c}^*)^{2}=z_{1-\alpha}^{2}$.  Note that
  $x_{i} \in \underline{\mathcal{X}}_t$ implies
  \begin{align*}
    \sum_{x_{i}\in\underline{\mathcal{X}}_t} \left[b_{t}'- C \left\lVert(x_{i})_{\mathcal{V}+} \right\rVert - C' \left\lVert(x_{i})_{\mathcal{V}-} \right\rVert \right]_{+}^{2} / \sigma^{2}(x_{i})
    = \sum_{x_{i}\in\underline{\mathcal{X}}_t} \left[b_{t}'- C' \left\lVert(x_{i})_{\mathcal{V}-} \right\rVert \right]_{+}^{2} / \sigma^{2}(x_{i}),
  \end{align*}
  using the assumption that $\mathcal{V} = \{1,...,d\}$. This implies that
  $b_{t}'(z_{1 - \alpha})$ is a function of the data and $C'$. Likewise, if we
  define $\overline{\mathcal{X}}_c := [0, \infty)^d \cap \mathcal{X}_c$, and
  $b_c' = b_c'(\delta_c)$ to be the solution to
  \begin{align*}
    \sum_{x_{i} \in \overline{\mathcal{X}}_c} 
    \left[b_{c}'- C' \left\lVert(x_{i})_{\mathcal{V}+} \right\rVert \right]_{+}^{2} 
    / \sigma^{2}(x_{i}) = \delta_c^2,
  \end{align*}
  $b_{c}'(z_{1 - \alpha})$ is a function of the data and $C'$. Therefore, we can
  set $A(C') = b_{t}'(z_{1 - \alpha}) + b_{c}'(z_{1 - \alpha})$.

  Next, let us consider necessity. First, suppose
  $(-\infty, 0]^d \cap \mathcal{X}_t = \emptyset$. Then, since every term in the
  summation in (\ref{eq: b_t proof of adpt cond}) depends on $C$, increasing $C$
  should necessarily increase the size of $b_t$. Therefore,
  $\omega(z_{1 - \alpha}, C, C')$ cannot be bounded by a term independent of
  $C$. The same reasoning applies to the case with
  $[0, \infty)^d \cap \mathcal{X}_c = \emptyset$, and to the case where
  $\mathcal{V} \subsetneq \{1,...,d\}$, which concludes the
  proof.\footnote{\label{foot: adpt dir proof} To be more precise, the reasoning
    above holds only if $\delta_t^*, \delta_c^* > 0$, and only if there is no
    observation exactly located on the axis. These cases can be ignored when $n$
    is sufficiently large, and when the running variables have a continuous
    distribution.}
\end{proof}
\begin{proof}[Proof of Theorem \ref{thm: minimax}]
  First, given some $\delta \geq 0$, consider the optimization problem
  \begin{eqnarray}
    \omega \left( \delta; C \right) & := & \underset{f_{1} \in \mathcal{F}(C), f_{2} \in \mathcal{F}(C)}{\sup}L_{RD} f_{2} - L_{RD} f_{1} \nonumber \\
                                    & \text{ } & \text{s.t.} \left\lVert\frac{f_{2}-f_{1}}{\sigma} \right\rVert_{2}^{2} \leq \delta^{2}. \label{eq:mod_of_cont-mm}
  \end{eqnarray}
  Denote the solutions to (\ref{eq:mod_of_cont-mm}) by
  $\left( f_{\delta, 1}^{*} , f_{\delta, 2}^{*} \right)$.  Define
  \begin{eqnarray*}
    \widetilde{L} (\delta) & := & \frac{\omega' \left( \delta; C, C \right)}{\delta}\sum_{i=1}^{n}\frac{\left( f_{\delta, 2}^{*}(x_{i}) - f_{\delta, 1}^{*}(x_{i}) \right)y_{i}}{\sigma^{2}(x_{i})} \\
                           &  & -\frac{\omega' \left( \delta; C, C \right)}{2\delta}\sum_{i=1}^{n} \left[\frac{\left( f_{\delta, 2}^{*}(x_{i}) - f_{\delta, 1}^{*}(x_{i}) \right)\left( f_{\delta, 1}^{*}(x_{i})+f_{\delta, 2}^{*}(x_{i}) \right)}{\sigma^{2}(x_{i})} \right]\\
                           &  & +L_{RD} \left( \frac{f_{\delta, 1}^{*}+f_{\delta, 2}^{*}}{2} \right).
  \end{eqnarray*}
  Under these definitions, We can show
  \begin{align*}
    \text{sd} \left(\widetilde{L} (\delta) \right) & = \omega' \left( \delta; C, C \right) \\
    \sup_{f \in \mathcal{F}(C)} {\text{bias}} \left(\widetilde{L} (\delta) \right) & = \frac{1}{2} \left(\omega \left( \delta; C, C \right)-\delta\omega' \left( \delta; C, C \right) \right).
  \end{align*}
  Define
  \[
    \widetilde{\chi} (\delta) = \text{cv}_{\alpha} \left( \frac{1}{2} \left(
        \frac{\sup_{f \in \mathcal{F}(C)} {\text{bias}} \left(\widetilde{L}
            (\delta) \right) }{\text{sd} \left(\widetilde{L} (\delta) \right)}
      \right) \right)\cdot\text{sd} \left(\widetilde{L} (\delta) \right).
  \]
  Then, the result of \citet{Donoho1994} implies that the minimax affine optimal
  CI is given by
  $\left[\widetilde{L} (\delta)-\widetilde{\chi} (\delta),\widetilde{L} (\delta)
    + \widetilde{\chi} (\delta)\right]$, when $\delta$ is chosen to minimize
  $\widetilde{\chi} (\delta)$. Thus.  the proof is done if we can show 1)
  $\widetilde{L} (\delta)$ has the same form as $\widehat{L}^{\text{mm}}$, and
  2) $\text{sd} \left(\widetilde{L} (\delta) \right)$ and
  $\sup_{f \in \mathcal{F}(C)} {\text{bias}} \left(\widetilde{L} (\delta)
  \right)$ and are given as \eqref{eq: sd and sup bias (implementation)}.

  For the first claim, recall that we can write for $j= 1, 2$
  \[
    f_{\delta, j}^{*}(x) =f_{\delta, j,t}^{*}(x) \mathbbm{1} \left\{ x \in
      \mathcal{X}_{t} \right\} +f_{\delta, j,c}^{*}(x) \mathbbm{1} \left\{ x \in
      \mathcal{X}_{c} \right\} .
  \]
  Therefore, we have
  \begin{eqnarray*}
    \widetilde{L} (\delta) & = & \frac{\omega' \left( \delta; C, C \right)}{\delta}\sum_{x_{i} \in \mathcal{X}_{t}}\frac{\left( f_{\delta, 2, t}^{*}(x_{i}) - f_{\delta, 1, t}^{*}(x_{i}) \right)y_{i}}{\sigma^{2}(x_{i})} \\
                           & = & -\frac{\omega' \left( \delta; C, C \right)}{\delta}\sum_{x_{i} \in \mathcal{X}_{c}}\frac{\left( f_{\delta, 1, c}^{*}(x_{i}) - f_{\delta, 2, c}^{*}(x_{i}) \right)y_{i}}{\sigma^{2}(x_{i})} \\
                           &  & -\frac{\omega' \left( \delta; C, C \right)}{2\delta}\sum_{x_{i} \in \mathcal{X}_{t}} \left[\frac{\left( f_{\delta, 2, t}^{*}(x_{i}) - f_{\delta, 1, t}^{*}(x_{i}) \right)\left( f_{\delta, 2, t}^{*}(x_{i})+f_{\delta, 1, t}^{*}(x_{i}) \right)}{\sigma^{2}(x_{i})} \right]\\
                           &  & +\frac{\omega' \left( \delta; C, C \right)}{2\delta}\sum_{x_{i} \in \mathcal{X}_{c}} \left[\frac{\left( f_{\delta, 1, c}^{*}(x_{i}) - f_{\delta, 2, c}^{*}(x_{i}) \right)\left( f_{\delta, 1, c}^{*}(x_{i})+f_{\delta, 2, c}^{*}(x_{i}) \right)}{\sigma^{2}(x_{i})} \right]\\
                           &  & +\frac{f_{\delta, 1, t}^{*}(0)+f_{\delta, 2, t}^{*}(0)}{2}-\frac{f_{\delta, 1, c}^{*}(0)+f_{\delta, 2, c}^{*}(0)}{2}.
  \end{eqnarray*}
  Defining
  \begin{eqnarray*}
    \widetilde{L}_{t} (\delta) & = & \frac{\omega' \left( \delta; C, C \right)}{\delta}\sum_{x_{i} \in \mathcal{X}_{t}}\frac{\left( f_{\delta, 2, t}^{*}(x_{i}) - f_{\delta, 1, t}^{*}(x_{i}) \right)y_{i}}{\sigma^{2}(x_{i})} \\
                               &  & -\frac{\omega' \left( \delta; C, C \right)}{2\delta}\sum_{x_{i} \in \mathcal{X}_{t}} \left[\frac{\left( f_{\delta, 2, t}^{*}(x_{i}) - f_{\delta, 1, t}^{*}(x_{i}) \right)\left( f_{\delta, 2, t}^{*}(x_{i})+f_{\delta, 1, t}^{*}(x_{i}) \right)}{\sigma^{2}(x_{i})} \right]\\
                               &  & +\frac{f_{\delta, 1, t}^{*}(0)+f_{\delta, 2, t}^{*}(0)}{2},
  \end{eqnarray*}
  and
  \begin{eqnarray*}
    \widetilde{L}_{c} (\delta) & = & \frac{\omega' \left( \delta; C, C \right)}{\delta}\sum_{x_{i} \in \mathcal{X}_{c}}\frac{\left( f_{\delta, 1, c}^{*}(x_{i}) - f_{\delta, 2, c}^{*}(x_{i}) \right)y_{i}}{\sigma^{2}(x_{i})} \\
                               &  & -\frac{\omega' \left( \delta; C, C \right)}{2\delta}\sum_{x_{i} \in \mathcal{X}_{c}} \left[\frac{\left( f_{\delta, 1, c}^{*}(x_{i}) - f_{\delta, 2, c}^{*}(x_{i}) \right)\left( f_{\delta, 1, c}^{*}(x_{i})+f_{\delta, 2, c}^{*}(x_{i}) \right)}{\sigma^{2}(x_{i})} \right]\\
                               &  & +\frac{f_{\delta, 1, c}^{*}(0)+f_{\delta, 2, c}^{*}(0)}{2},
  \end{eqnarray*}
  we can see
  $\widetilde{L} (\delta) = \widetilde{L}_{t} (\delta)-\widetilde{L}_{c}
  (\delta)$ holds.

  Now, using the equations (\ref{eq: f2t - f1t}), (\ref{eq:f2t - f1t-0}) and
  (\ref{eq:f2t+f1c}) in Lemma \ref{lem: f2 - f1 _ f2 + f1}, and the first
  representation of $\frac{\omega' \left( \delta; C, C \right)}{\delta}$ in
  Lemma \ref{lem: om_prime}, we get
  \begin{eqnarray*}
    \widetilde{L}_{t} (\delta) & = & \frac{\sum_{x_{i} \in \mathcal{X}_{t}} \left[ \omega_{t} \left( \delta_{t}^{*}; C, C \right)-C \left\lVert (x)_{\mathcal{V}+} \right\rVert -C \left\lVert (x)_{\mathcal{V}-} \right\rVert \right]_{+}y_{i} / \sigma^{2}(x_{i})}{\sum_{x_{i} \in \mathcal{X}_{t}} \left[ \omega_{t} \left( \delta_{t}^{*}; C, C \right)-C \left\lVert (x)_{\mathcal{V}+} \right\rVert -C \left\lVert (x)_{\mathcal{V}-} \right\rVert \right]_{+} / \sigma^{2}(x_{i})} \\
                               &  & -\frac{1}{2}\sum_{x_{i} \in \mathcal{X}_{t}} \left[\frac{\left[ \omega_{t} \left( \delta_{t}^{*}; C, C \right)-C \left\lVert (x)_{\mathcal{V}+} \right\rVert -C \left\lVert (x)_{\mathcal{V}-} \right\rVert \right]_{+}}{\sum_{x_{i} \in \mathcal{X}_{t}} \left[ \omega_{t} \left( \delta_{t}^{*}; C, C \right)-C \left\lVert (x)_{\mathcal{V}+} \right\rVert -C \left\lVert (x)_{\mathcal{V}-} \right\rVert \right]_{+} / \sigma^{2}(x_{i})} \right.\\
                               &  & \left.\times\left(\omega_{t} \left( \delta_{t}^{*}; C, C \right) + C \left\lVert (x)_{\mathcal{V}+} \right\rVert -C \left\lVert (x)_{\mathcal{V}-} \right\rVert\right) / \sigma^{2}(x_{i})\right]\\
                               &  & +\frac{\omega_{t} \left( \delta_{t}^{*}; C, C \right)}{2}, \\
                               &=&
                                   \frac{\sum_{x_i \in \mathcal{X}_t} K(x_i/h_t) y_i / \sigma^2(x_i)}{\sum_{x_i \in \mathcal{X}_t} K(x_i/h_t) /  \sigma^2(x_i)} + a_t(\delta).
  \end{eqnarray*}
  Similarly, for $\widetilde{L}_{c} (\delta)$, using the equations (\ref{eq :f1c
    - f2c}), (\ref{eq:f2c-f1c-0}), and (\ref{eq:f2c+f1c}) in Lemma \ref{lem: f2
    - f1 _ f2 + f1}, and the second representation of
  $\frac{\omega' \left( \delta; C, C \right)}{\delta}$ in Lemma \ref{lem:
    om_prime}, we can get
  \begin{align*}
    \widetilde{L}_{c} (\delta) = \frac{\sum_{x_i \in \mathcal{X}_c} K(x_i/h_c) y_i / \sigma^2(x_i)}{\sum_{x_i \in \mathcal{X}_c} K(x_i/h_c) / \sigma^2(x_i)} + a_c(\delta),
  \end{align*}
  which establishes the claim that $\widetilde{L} (\delta)$ has the same form as
  $\widehat{L}^{\text{mm}}$.

  Next, for the second claim about the standard deviation and the worst-case
  bias, it is sufficient to show
  \begin{eqnarray*}
    \omega' \left( \delta; C, C \right) & = & \frac{\delta / (C             h_{t}(\delta))}
                                              {\sum_{x_{i} \in \mathcal{X}_{t}} K\left(x_{i} / h_{t}(\delta) \right) / \sigma^2(x_i)}, \\
    \omega \left( \delta; C, C \right) & = & C(h_{t}(\delta) + h_{c}(\delta)).
  \end{eqnarray*}
  First, the first equation follows from Lemma \ref{lem: om_prime}.  Next,
  $\omega \left( \delta; C, C \right) = \omega_{t} \left( \delta_{t}^{*}(\delta;
    C); C \right) + \omega_{c} \left( \delta_{c}^{*}(\delta; C); C \right)$
  follows from Lemma \ref{lem:Mod_RD_form}, which implies the second equation,
  which concludes the proof.
\end{proof}
\begin{proof}[Proof of Theorem \ref{thm: adpt-CI}]

  Given some $f\in \mathcal{F}(C)$ and $j \in \{1,...,J\}$, we can write
  \begin{eqnarray*}
    P_{f} \left(L_{RD} f<\hat{c}_{\tau}^L(C_j) \right) & = & P_{f} \left(\hat{c}_{\tau}^L(C_j)-L_{RD} f>0\right) \\
                                                       & = & P_{f} \left( \frac{\hat{c}_{\tau}^L(C_j)-L_{RD} f}{\omega' \left(z_{1-\tau}, C, C_{j} \right)}>0\right) \\
                                                       & = & P_{f} \left( \frac{\hat{c}_{\tau}^L(C_j)-L_{RD} f}{\omega' \left(z_{1-\tau}, C, C_{j} \right)}+z_{1-\tau}>z_{1-\tau} \right) \\
                                                       & \equiv & P_{f} \left(\widetilde{V}_{j, \tau}>z_{1-\tau} \right),
  \end{eqnarray*}
  where $\omega' \left(z_{1-\tau}, C, C_{j} \right)$ is as defined in Therefore,
  defining
  $CI^{L}(\tau):=\left[\underset{1\leq j\leq
      J}{\max}\hat{c}_{\tau}^L(C_j),\infty\right)$, we have
  \begin{eqnarray*}
    P_{f} \left(L_{RD} f\notin CI^{L}(\tau) \right) & = & P_{f} \left(\exists j\text{ s.t. }L_{RD} f<\hat{c}_{\tau}^L(C_j) \right) \\
                                                    & = & P_{f} \left(\underset{1\leq j\leq J}{\text{max}}\widetilde{V}_{j, \tau}>z_{1-\tau} \right).
  \end{eqnarray*}
  Now, we want to find the largest value of $\tau$ such that
  \begin{equation}
    \label{eq: v tilde max cov}
    \underset{f\in \mathcal{F}(C)}{\text{sup}}P_{f} \left(\underset{1\leq j\leq J}{\text{max}}\widetilde{V}_{j, \tau}>z_{1-\tau} \right)\leq\alpha,
  \end{equation}
  therefore satisfying the coverage requirement. First, it easy to see that
  $\left(\widetilde{V}_{1, \tau},...,\widetilde{V}_{J, \tau} \right)$ has a
  multivariate normal distribution. Next, note that the quantiles of
  $\underset{1\leq j\leq J}{\text{max}}\widetilde{V}_{j, \tau}$ is increasing in
  each of $E_{f}\widetilde{V}_{1, \tau},...,E_{f}\widetilde{V}_{j, \tau}$.
  Moreover, the variances and covariances of
  $\left(\widetilde{V}_{1, \tau},...,\widetilde{V}_{J, \tau} \right)$ do not
  depend on the true regression function $f$, by the construction of
  $\hat{c}_{\tau}^L(C_j)$. This means that the quantiles of
  $\underset{1\leq j\leq J}{\text{max}}\widetilde{V}_{j, \tau}$ are smaller than
  those of $\underset{1\leq j\leq J}{\text{max}}V_{j, \tau}$, where
  $\left(V_{1, \tau},...,V_{J, \tau} \right)$ has a multivariate normal
  distribution with a mean given by
  $EV_{j, \tau} = \underset{f\in
    \mathcal{F}(C)}{\text{sup}}E_{f}\widetilde{V}_{j, \tau}$ for $j = 1,...,J$,
  and with the same covariance matrix as
  $\left(\widetilde{V}_{1, \tau},...,\widetilde{V}_{J, \tau} \right)$.
  Therefore, if we take $\tau$ so that $z_{1-\tau}$ is the $1-\alpha$th quantile
  of $\underset{1\leq j\leq J}{\text{max}}V_{j, \tau}$, (\ref{eq: v tilde max
    cov}) is satisfied. Therefore, it remains to show
  $\left( V_{j, \tau} \right)_{j = 1}^J$ is a multivariate normal distribution
  with zero means, unit variances, and covariances as given in (\ref{eq: Cov V_j
    tau}).

  For the mean, by definition of $\hat{c}_{\tau}^L(C_j)$, we have
  \begin{align*}
    \underset{f\in \mathcal{F}(C)}{\text{sup}}E_{f} \left[\hat{c}_{\tau}^L(C_j)-L_{RD} f\right] &=  -z_{1-\tau}\text{sd}(\widehat{L}_\tau(C_j)) \\
                                                                                                &= -z_{1-\tau}\omega'(z_{1-\tau}, C, C_j),
  \end{align*}
  where the last line follows from the discussion in
  \cite{armstrong2018optimal}. Hence, we get $EV_{j, \tau} = 0$. Moreover, this
  also implies $\text{Var}(\widehat{L}_\tau(C_j)) = 1$.
  For the covariance, we have
  \begin{align*}
    & \text{Cov} \left(\widehat{L}_{\tau}(C_j),\widehat{L}_{\tau}(C_k) \right) \\
    = & \text{Cov} \left(\widehat{L}_{t, \tau}(C_j)-\widehat{L}_{c, \tau}(C_j),\widehat{L}_{t, \tau}(C_k)-\widehat{L}_{c, \tau}(C_k) \right) \\
    = & \text{Cov} \left(\widehat{L}_{t, \tau}(C_j),\widehat{L}_{t, \tau}(C_k) \right) + \text{Cov} \left(\widehat{L}_{c, \tau}(C_j),\widehat{L}_{c, \tau}(C_k) \right),
  \end{align*}
  where the last line follows from the independence assumption, and we can show
  \begin{align*}
    & \text{Cov} \left( \frac{\widehat{L}_{t, \tau}(C_j)}{\omega' \left(z_{1-\tau}, C, C_{j} \right)},\frac{\widehat{L}_{t, \tau}(C_k)}{\omega' \left(z_{1-\tau}, C, C_{k} \right)} \right) \\
    = & \frac{\sum_{x_{i} \in \mathcal{X}_{t}}
        K \left(x_{i}, h_{t, \tau}(C_j) \right)
        K \left(x_{i}, h_{t, \tau}(C_k) \right) /
        \sigma^{2}(x_{i})}
        {\sum_{x_{i} \in \mathcal{X}_{t}}
        K \left(x_{i}, h_{t, \tau}(C_j) \right) /
        \sigma^{2}(x_{i})
        \sum_{x_{i} \in \mathcal{X}_{t}}
        K \left(x_{i}, h_{t, \tau}(C_k) \right) /
        \sigma^{2}(x_{i})} \\
    & \times\frac{1}{\omega' \left(z_{1-\tau}, C, C_{j} \right)\omega' \left(z_{1-\tau}, C, C_{k} \right)} \\
    = & \frac{\omega_{t}^*(z_{1-\tau}; C, C_j) \omega_{t}^*(z_{1-\tau}; C, C_k)}{z_{1-\tau}^{2}}\sum_{x_{i} \in \mathcal{X}_{t}}K \left(x_{i}, h_{t, \tau}(C_j) \right)K \left(x_{i}, h_{t, \tau}(C_k) \right) / \sigma^{2}(x_{i}),
  \end{align*}
  where the last equality follows from Lemma \ref{lem: om_prime} and the
  definition of $K \left(x_{i}, h_{t, \tau}(C_j) \right)$.  Likewise, we can
  show
  \begin{align*}
    & \text{Cov} \left(\widehat{L}_{c, \tau}(C_j),\widehat{L}_{c, \tau}(C_k) \right) \\
    = & \frac{\omega_{c}^*(z_{1-\tau}; C_j, C) \omega_{c}^*(z_{1-\tau}; C_k, C)}{z_{1-\tau}^{2}}\sum_{x_{i} \in \mathcal{X}_{t}}K \left(x_{i}, h_{c, \tau}(C_j) \right)K \left(x_{i}, h_{c, \tau}(C_k) \right) / \sigma^{2}(x_{i}),
  \end{align*}
  which proves that the covariance term has the same form as in (\ref{eq: Cov
    V_j tau}).
\end{proof}
\begin{proof}[Proof of Proposition \ref{prop: ell^adpt calc}]
  The last statement is immediate from Lemma \ref{lem:Mod_RD_form} of our paper
  and Theorem 3.1 of \cite{armstrong2018optimal}, so we focus on the former
  statement.
  
  First, we state the means and the covariance matrix of $U_j$'s. Given $\mathcal{C} = \left\{C_1,...,C_J\right\}$, let
  $\tau^* = \tau^*(\mathcal{C})$ be the value of $\tau$ that solves \eqref{eq:
    tau* def}. Then, we have
  \begin{align*}
    EU_j =
    &\frac{\sum_{x_{i}\in\mathcal{X}_{t}}
      K(x_i, h_{t, \tau^*}(C_j))C'\left|\left|(x_i)_{\mathcal{V}-}\right|\right|/
      \sigma^2(x_i)}
      {\sum_{x_{i}\in\mathcal{X}_{t}}
      K(x_i, h_{t, \tau^*}(C_j))/
      \sigma^2(x_i)} \\
    &+
      \frac{\sum_{x_{i}\in\mathcal{X}_{c}}
      K(x_i, h_{c, \tau^*}(C_j))C'\left|\left|(x_i)_{\mathcal{V}+}\right|\right|/
      \sigma^2(x_i)}
      {\sum_{x_{i}\in\mathcal{X}_{c}}
      K(x_i, h_{c, \tau^*}(C_j))/
      \sigma^2(x_i)}\\
    &+
      z_{1 - \tau^*} \text{\emph{sd}} ( \widehat{L}_{\tau^*}(C_j) )
      +
      \sup_{f \in \mathcal{F}(C)} \text{\emph{bias}} (
      \widehat{L}_{\tau^*}(C_j)
      ),
  \end{align*}
  and
  \begin{align*}
    \text{\emph{Cov}}(U_j,U_k) =
    &\frac{\sum_{x_{i}\in\mathcal{X}_{t}}
      K(x_i, h_{t, \tau^*}(C_j))
      K(x_i, h_{t, \tau^*}(C_k))
      /
      \sigma^2(x_i)}
      {\sum_{x_{i}\in\mathcal{X}_{t}}
      K(x_i, h_{t, \tau^*}(C_j))/
      \sigma^2(x_i)
      \sum_{x_{i}\in\mathcal{X}_{t}}
      K(x_i, h_{t, \tau^*}(C_k))/
      \sigma^2(x_i)
      } \\
    &+
      \frac{\sum_{x_{i}\in\mathcal{X}_{c}}
      K(x_i, h_{c, \tau^*}(C_j))
      K(x_i, h_{c, \tau^*}(C_k))/
      \sigma^2(x_i)}
      {\sum_{x_{i}\in\mathcal{X}_{c}}
      K(x_i, h_{c, \tau^*}(C_j))/
      \sigma^2(x_i)
      \sum_{x_{i}\in\mathcal{X}_{c}}
      K(x_i, h_{c, \tau^*}(C_k))/
      \sigma^2(x_i)}.
  \end{align*}

  To prove the statement in the proposition, define
  \[
    f^*_{C'}(x):=-C'\left|\left|(x)_{\mathcal{V}-}\right|\right|1(x\in\mathcal{X}_{t})+C'\left|\left|(x)_{\mathcal{V}+}\right|\right|1(x\in\mathcal{X}_{c}).
  \]
  Then, we show
  \begin{eqnarray*}
    \ell^{\text{adpt}}(C'; \mathcal{C})
    & = & 
          \underset{f\in\mathcal{F}(C')}
          {\sup}
          E_{f}\underset{1 \leq j \leq J}{\min}\left(L_{RD}f-\hat{c}_{\tau^*}^L(C_j)\right)\\
    & =: & \underset{f\in\mathcal{F}(C')}
           {\sup}
           E_{f}\underset{1 \leq j \leq J}{\min}\left(Z_{j}(f)\right)\\
    & = & 
          E_{f^*_{C'}}
          \underset{1 \leq j \leq J}{\min}\left(Z_{j}\left(f^*_{C'}\right)\right) \\
    & = & 
          E_{f^*_{C'}}
          \underset{1 \leq j \leq J}{\min} -\hat{c}_{\tau^*}^L(C_j).
  \end{eqnarray*}
  The first and the last equalities hold by definition (since
  $L_{RD}f^*_{C'} = 0$), so what we have to show is the second to last equality.

  It is sufficient to show that for each $j\in \{1,...,J\}$,
  \begin{eqnarray*}
    \underset{f\in\mathcal{F}(C')}{\sup}E_{f}\left(L_{RD}f-\hat{c}_{\tau^*}^L(C_j)\right)
    &=&
        E_{f^*_{C'}}\left(L_{RD}f^*_{C'}-\hat{c}_{\tau^*}^L(C_j)\right)\\
    &=&
        E_{f^*_{C'}}\left(\hat{c}_{\tau^*}^L(C_j)\right).
  \end{eqnarray*}
  Note that $\left( Z_j (f) \right)_{j = 1}^J$ has a multivariate normal
  distribution with some mean vector $\left( \mu_{j}(f) \right)_{j = 1}^J$ and a
  variance matrix which does not depend on $f$. This implies if $f^*_{C'}$
  maximizes $\mu_{j}(f)$ for all $j = 1,...,J$, it also maximizes
  $E_{f} \; \underset{1 \leq j \leq J}{\min}\left(Z_{j}\left(f\right)\right)$.

  For a shorthand notation, define
  \begin{align*}
    \chi^L_\tau(C') := 
    \sup_{f \in \mathcal{F}(C)} \text{bias}_f \left( \widehat{L}_\tau(C') \right)
    +
    z_{1-\tau}\text{sd} \left(\widehat{L}_\tau(C') \right).
  \end{align*}
  Now, note that we can write
  \begin{align*}
    \underset{f\in\mathcal{F}(C')}{\sup}E_{f}\left(L_{RD}f-\hat{c}_{\tau^*}^L(C_j)\right)
    &=
      \underset{f\in\mathcal{F}(C')}{\sup}E_{f}\left(L_{RD}f-\widehat{L}_{\tau^*}(C_j)+\chi^L_{\tau^*}(C_j)\right) \\
    &=
      \underset{f\in\mathcal{F}(C')}{\sup}E_{f}\left(L_{RD}f-\widehat{L}_{\tau^*}(C_j)\right)
      + \chi^L_{\tau^*}(C_j),
  \end{align*}
  since $\chi^L_{\tau^*}(C_j)$ is a fixed quantity.  The first term in the last
  line can be written as
  \begin{align*}
    & \underset{f\in\mathcal{F}(C')}{\sup}E_{f}\left(L_{RD}f-\widehat{L}_{\tau^*}(C_j)\right) \\
    =\ &  \underset{f_{t,},f_{c}\in\mathcal{F}(C')}{\sup}E_{f}\left[\left(f_{t}(0)-\frac{\sum_{x_{i}\in\mathcal{X}_{t}}K(x_i, h_{t, \tau^*}(C_j))y_{i} / \sigma^2(x_i)}{\sum_{x_{i}\in\mathcal{X}_{t}}K(x_i, h_{t, \tau^*}(C_j)) / \sigma^2(x_i)}\right)\right. \\
    &-
      \left.\left(f_{c}(0)-\frac{\sum_{x_{i}\in\mathcal{X}_{c}}K(x_i, h_{c, \tau^*}(C_j))y_{i} / \sigma^2(x_i)}{\sum_{x_{i}\in\mathcal{X}_{c}}K(x_i, h_{c, \tau^*}(C_j)) / \sigma^2(x_i)}\right)\right]\\
    =\ &   \underset{f_{t}\in\mathcal{F}(C')}{\sup}E_{f_{t}}\left(f_{t}(0)-\frac{\sum_{x_{i}\in\mathcal{X}_{t}}K(x_i, h_{t, \tau^*}(C_j))f_{t}(x_{i}) / \sigma^2(x_i)}{\sum_{x_{i}\in\mathcal{X}_{t}}K(x_i, h_{t, \tau^*}(C_j)) / \sigma^2(x_i)}\right)\\ 
    &-
      \underset{f_{c}\in\mathcal{F}(C')}{\text{inf}}E_{f_{c}}\left(f_{c}(0)-\frac{\sum_{x_{i}\in\mathcal{X}_{c}}K(x_i, h_{c, \tau^*}(C_j))f_{c}(x_{i}) / \sigma^2(x_i)}{\sum_{x_{i}\in\mathcal{X}_{c}}K(x_i, h_{c, \tau^*}(C_j)) / \sigma^2(x_i)}\right).
  \end{align*}
  Note that $f_{t}(0)$ and $f_{c}(0)$ can be normalized to 0 since if we
  consider $\widetilde{f}_{t}=f_{t}+v$ for some $v\in\mathbb{R}$,
  \begin{align*}
    &\widetilde{f}_{t}(0)-\frac{\sum_{x_{i}\in\mathcal{X}_{t}}K(x_i, h_{t, \tau^*}(C_j))\widetilde{f}_{t}(x_{i}) / \sigma^2(x_i)}{\sum_{x_{i}\in\mathcal{X}_{t}}K(x_i, h_{t, \tau^*}(C_j)) / \sigma^2(x_i)} \\
    =\ & f_{t}(0)+v-\frac{\sum_{x_{i}\in\mathcal{X}_{t}}K(x_i, h_{t, \tau^*}(C_j) / \sigma^2(x_i))\left(f_{t}(x_{i})+v\right) / \sigma^2(x_i)}{\sum_{x_{i}\in\mathcal{X}_{t}}K(x_i, h_{t, \tau^*}(C_j)) / \sigma^2(x_i)}\\
    =\ & f_{t}(0)-\frac{\sum_{x_{i}\in\mathcal{X}_{t}}K(x_i, h_{t, \tau^*}(C_j))f_{t}(x_{i}) / \sigma^2(x_i)}{\sum_{x_{i}\in\mathcal{X}_{t}}K(x_i, h_{t, \tau^*}(C_j)) / \sigma^2(x_i)},
  \end{align*}
  and likewise for $f_{c}$.  Therefore, we have
  \begin{align*}
    &\underset{f_{t}\in\mathcal{F}(C')}{\sup}E_{f_{t}}\left(f_{t}(0)-\frac{\sum_{x_{i}\in\mathcal{X}_{t}}K(x_i, h_{t, \tau^*}(C_j))f_{t}(x_{i}) / \sigma^2(x_i)}{\sum_{x_{i}\in\mathcal{X}_{t}}K(x_i, h_{t, \tau^*}(C_j)) / \sigma^2(x_i)}\right) \\
    =\ &
         - \underset{f_{t}\in\mathcal{F}(C'),f_{t}(0)=0}{\inf}\frac{\sum_{x_{i}\in\mathcal{X}_{t}}K(x_i, h_{t, \tau^*}(C_j))f_{t}(x_{i}) / \sigma^2(x_i)}{\sum_{x_{i}\in\mathcal{X}_{t}}K(x_i, h_{t, \tau^*}(C_j)) / \sigma^2(x_i)}.
  \end{align*}
  \cite{kwonkwon2019regpoint} show that the minimizer to this problem is given
  by
  \[
    f_{t}^*(x)=-C'\left|\left|(x)_{\mathcal{V}-}\right|\right|.
  \]
  Likewise, we have
  \begin{align*}
    & - \underset{f_{c}\in\mathcal{F}(C')}{\text{inf}}E_{f_{c}}\left(f_{c}(0)-\frac{\sum_{x_{i}\in\mathcal{X}_{c}}K(x_i, h_{c, \tau^*}(C_j))f_{c}(x_{i}) / \sigma^2(x_i)}{\sum_{x_{i}\in\mathcal{X}_{c}}K(x_i, h_{c, \tau^*}(C_j)) / \sigma^2(x_i)}\right)\\
    =\ &
         \underset{f_{c}\in\mathcal{F}(C'),f_{c}(0)=0}{\text{sup}}\frac{\sum_{x_{i}\in\mathcal{X}_{c}}K(x_i, h_{c, \tau^*}(C_j))f_{c}(x_{i}) / \sigma^2(x_i)}{\sum_{x_{i}\in\mathcal{X}_{c}}K(x_i, h_{c, \tau^*}(C_j)) / \sigma^2(x_i)}.
  \end{align*}
  Again by \cite{kwonkwon2019regpoint}, the maximizer is given by
  \[
    f_{c}^*(x)=C'\left|\left|(x)_{\mathcal{V}+}\right|\right|.
  \]
  Therefore, the worst case expected length is achieved when $f =
  f^*_{C'}$. Lastly, it is easy to see that $- \hat{c}_j^L$'s are jointly
  normally distributed with the mean and the covariance given in the lemma when
  $f = f^*_{C'}$, which establishes our claim.
\end{proof}

\section{Unbiased estimator for $\underline{C}$}
\label{sec: minC-est}

As in \cite{armstrong2018optimal_supp}, we can estimate the lower bound of $C$
using the data. We first explain the possibility for the case with
$d=1$. Suppose $\mathcal{X}_t = \left[ x_{\min},0 \right]$ and
$\mathcal{X}_c = \left(0, x_{\max} \right]$.

First, note that for any $x_2\geq x_1\geq 0$, we have
\begin{equation}
  f_c\left(x_2\right) - f_c\left(x_1 \right) \leq C \left(x_2 - x_1\right).
\end{equation}Let $n_c := \sum_{i=1}^n \mathbbm{1} \left\{x_i \in \mathcal{X}_c \right\}$, and set $a_m^c>0$ such that $\sum_{i=1}^n \mathbbm{1} \left\{x_i \in \mathcal{X}_c,\hspace{2pt} x_i \leq a_m^c \right\} = n_c/2$. Then, we have 
\begin{align}
  &\frac{1}{n_c/2} \left[\sum_{i=1}^n f_c\left(x_i\right) \mathbbm{1} \left\{x_i \in \mathcal{X}_c,\hspace{2pt} x_i > a_m^c \right\} - 
    \sum_{i=1}^n f_c\left(x_i\right) \mathbbm{1} \left\{x_i \in \mathcal{X}_c,\hspace{2pt} x_i \leq  a_m^c \right\} \right] \\
  \leq & \frac{C}{n_c/2}  \sum_{i=1}^n x_i \mathbbm{1} \left\{x_i \in \mathcal{X}_c,\hspace{2pt} x_i > a_m^c \right\},
\end{align}so that we have 
\begin{equation}
  \label{C-lbd-formula}
  C \geq 
  \frac{\frac{1}{n_c/2} \left[\sum_{i=1}^n f_c\left(x_i\right) \mathbbm{1} \left\{x_i \in \mathcal{X}_c,\hspace{2pt} x_i > a_m^c \right\} - 
      \sum_{i=1}^n f_c\left(x_i\right) \mathbbm{1} \left\{x_i \in \mathcal{X}_c,\hspace{2pt} x_i \leq  a_m^c \right\} \right] }
  {\frac{1}{n_c/2} \left[
      \sum_{i=1}^n x_i \mathbbm{1} \left\{x_i \in \mathcal{X}_c,\hspace{2pt} x_i > a_m^c \right\} - \sum_{i=1}^n x_i \mathbbm{1} \left\{x_i \in \mathcal{X}_c,\hspace{2pt} x_i \leq a_m^c \right\} \right]
  }.
\end{equation}
We estimate the RHS of (\ref{C-lbd-formula}), denoted by $\mu_c$, by
\begin{equation}
  \hat{\mu}_c := \frac{\frac{1}{n_c/2} \left[\sum_{i=1}^n y_i \mathbbm{1} \left\{x_i \in \mathcal{X}_c,\hspace{2pt} x_i > a_m^c \right\} - 
      \sum_{i=1}^n y_i \mathbbm{1} \left\{x_i \in \mathcal{X}_c,\hspace{2pt} x_i \leq  a_m^c \right\} \right] }
  {\frac{1}{n_c/2} \left[
      \sum_{i=1}^n x_i \mathbbm{1} \left\{x_i \in \mathcal{X}_c,\hspace{2pt} x_i > a_m^c \right\} - 
      \sum_{i=1}^n x_i \mathbbm{1} \left\{x_i \in \mathcal{X}_c,\hspace{2pt} x_i \leq a_m^c \right\}
    \right]}.
\end{equation}We can easily see that $\hat{\mu}_c$  is an unbiased estimator of $\mu_c$. 

Likewise, we can form a lower bound estimator using the data in $\mathcal{X}_t$
by
\begin{equation}
  \hat{\mu}_t := \frac{\frac{1}{n_t/2} \left[\sum_{i=1}^n y_i \mathbbm{1} \left\{x_i \in \mathcal{X}_t,\hspace{2pt} x_i > a_m^t \right\} - 
      \sum_{i=1}^n y_i \mathbbm{1} \left\{x_i \in \mathcal{X}_t,\hspace{2pt} x_i \leq  a_m^t \right\} \right] }
  {\frac{1}{n_t/2} \left[
      \sum_{i=1}^n x_i \mathbbm{1} \left\{x_i \in \mathcal{X}_t,\hspace{2pt} x_i > a_m^t \right\}-
      \sum_{i=1}^n x_i \mathbbm{1} \left\{x_i \in \mathcal{X}_t,\hspace{2pt} x_i \leq a_m^t \right\}
    \right]},
\end{equation}where  $n_t := \sum_{i=1}^n \mathbbm{1} \left\{x_i \in \mathcal{X}_t \right\}$ and $a_m^t<0$ was set so that $\sum_{i=1}^n \mathbbm{1} \left\{x_i \in \mathcal{X}_t,\hspace{2pt} x_i \leq a_m^t \right\} = n_t/2$. In the  dataset of \cite{lee2008randomized}, it turns out that $\hat{\mu}_c = 0.353$ and $\hat{\mu}_t=0.355$.  

We can use similar reasoning for the case with $d>1$. For example, when $d=d_+$,
we use inequality
\begin{equation}
  f_c\left(x_2\right) - f_c\left(x_1\right) \leq C \left\lVert x_2 - x_1 \right \rVert, 
\end{equation}for any $x_1,x_2$ such that $x_{2r} \geq x_{1r}$ for all $r=1,...,d$. Find two subsets of $\mathcal{X}_c$ , say $\mathcal{X}_{c1}$  and $\mathcal{X}_{c2}$  so that 1) $x_2 \in \mathcal{X}_{c2}$  and  $x_1 \in \mathcal{X}_{c1}$ implies  $x_{2r} \geq x_{1r}$ for all $r=1,...,d$, and 2) the number of observations in each set is equal to $\tilde{n}_c/2$. Index $x_i$'s so that $x_1,...,x_{\tilde{n}_c/2}\in \mathcal{X}_{c1}$ and $x_{\tilde{n}_c/2+1},...,x_{\tilde{n}_c}\in \mathcal{X}_{c2}$. Define some one-to-one mapping $j$ from $\left\{1,2,...,\frac{\tilde{n}_c}{2} \right\}$ to $\left\{\frac{\tilde{n}_c}{2}+1,\frac{\tilde{n}_c}{2}+2,...,\tilde{n}_c \right\}.$ Then, we have
\begin{equation}
  C \geq \frac{\sum_{i=1}^{\tilde{n}_c/2} \left[
      f_c\left(x_{j(i)} \right) - f_c\left(x_i\right)\right]}
  {\sum_{i=1}^{\tilde{n}_c/2} \left\lVert
      x_{j(i)} -x_i \right\rVert}.
\end{equation}Again, the lower bound can be estimated by replacing $f_c\left(x_i\right)$  by $y_i$. 
\section{Examples of Multi-score RDD}
\label{sec: multi RD monotone}

In this section, we discuss RDD applications with monotone multiple running
variables.\footnote{Refer to Appendix A.3 of \cite{babii2019isotonic} for RDD applications with univariate monotone running variables} When presenting these empirical applications, we categorize the class
of RD designs we consider into the following three cases, depending on how the
running variables relate to the assignment of the treatment. While our framework
can potentially cover a much larger class of models, these three settings seem
to be the ones that appear most frequently in empirical studies.\footnote{For
  example, our categorization excludes geographic RD designs, such as those
  analyzed in \citet{dell2010persistent}, \citet{keele2015geographic}, and
  \cite{imbens2019optimized}. While our general framework can incorporate such
  models, we do not consider them since the monotonicity assumption is unlikely
  to hold under such contexts.} We let $x_i \in \mathbb{R}^d$ be the value of
running variable for an individual $i$, and denote the $s$th element of $x_i$ by
$x_{i(s)}$ for $s = 1,...,d$.

\begin{enumerate}
\item \textbf{MRO (Multiple Running variables with ``OR'' conditions):} An
  individual $i$ is treated if there exists some
  $s \in \left\{ 1,...,d\right\} $ such that $x_{i(s)}>0$ (or $\leq0$).
\item \textbf{MRA (Multiple Running variables with ``AND'' conditions):} An
  individual $i$ is treated if $x_{i(s)}>0$ (or $\leq0$) for all $s = 1,...,d$.
\item \textbf{WAV (Weighted AVerage of multiple running variables):} There are
  multiple running variables, and the treatment status is determined by some
  weighted average of those running variables.  Hence, an individual $i$ is
  treated if $\sum_{s=1}^{d}w_{s}x_{i(s)}>0$ (or $\leq0$) for some positive
  weights $w_{1},...,w_{d}$. While we can view this design as an RD design with a
  single running variable $\widetilde{x}_{i} = \sum_{s=1}^{d}w_{s}x_{i(s)}$, we
  may obtain richer information by considering this design as a
  multi-dimensional RD design. For example, consider an RD design where an
  individual $i$ is treated if the average of the normalized math and reading
  scores is greater than 0. Then, we can consider treatment effect parameters at
  different cutoff points, e.g. (math $=0$, reading $=0$), (math $=-1$, reading
  $=1$), or (math $=1$, reading $=-1$).
\end{enumerate}
Below we list empirical examples which correspond to the MRO, MRA or WAV
settings when $d > 1$.

\citet{jacob2004remedial} and \citet{matsudaira2008mandatory} consider the
effect of mandatory summer school on academic achievements in later years. In
their empirical context, summer school attendance is required if a student's
math score or reading score is below a certain threshold, and thus this example
corresponds to the MRO case. The outcome variable is the math or the reading
score in the next year. It seems plausible to assume any test score in the next
year is increasing on average in the test scores of former years, so we can
impose monotonicity in this example. Alternatively, the monotone relationship
between different test subjects (e.g., the reading score in the previous exam as
a running variable and the math score in next exam as an outcome variable) might
be questionable, in which case we can impose a partial monotonicity restriction.
 
\cite{papay2010consequences} consider the effect of a student's failing a high
school exit exam (Massachusetts Comprehensive Assessment System) on the
probability of high school graduation. Since there are two portions of the exam,
mathematics and English, this setting can be also viewed as the MRO case. It
seems reasonable to assume that the graduation probability is increasing in the
exam scores, so we may impose monotonicity in this example as well.
 
\citet{kane2003quasi} investigates the impact of the college price subsidy on
college enrollment decision. To be eligible for the subsidy in one such program
run by California called ``Cal Grant A'', the applicant's high school GPA had to
be above a certain cutoff while her income and assets had be below certain
cutoffs. Hence, this example falls into the MRA case. It seems plausible that
the college enrollment probability is increasing in the high school GPA, while
it can be debatable whether the income and the asset levels have monotone
relationships with the outcome. So in this example we may assume either full or
partial monotonicity depending on empirical researchers' belief.
 
\citet{van2002estimating} analyzes the effect of financial aid offers on
students' college enrollment decisions. In that paper, an East Coast college
offers a financial aid if a student's weighted average of SAT score and GPA
exceeds some threshold level, so this example corresponds to the WAV setting. It
seems also plausible to assume the enrollment probability is increasing on
average with respect to the SAT and GPA scores.

\citet{chay2005central} investigate the effect of a government program that
allocated specific resources on the academic performance of schools. The
resource was allocated to schools whose average of mathematics and language
scores of their students falls below some cutoff point, corresponding to the WAV
setting. It seems reasonable to assume that academic scores of a school are
increasing in the previous test scores of its students on average.

\cite{leuven2007effect} evaluate the effect of government subsidies to schools
on academic achievements. In Netherlands, the subsidy was provided to schools
with more than 70\% of disadvantaged pupils, where the proportion is calculated
as the sum of ethnic minority students and students with parents whose education
level is no more than secondary school. So this example corresponds to the WAV
setting. If we can argue that the academic achievement metrics of a school are
decreasing in its proportion of disadvantaged pupils, we may impose the
monotonicity assumption.

For more examples of RDD with multiple running variables, the reader is referred
to \cite{papay2011extending} and \citet{wong2013analyzing}. We also note that
all of the three settings discussed above become the standard single-dimensional
RD design when $d = 1$.
\newpage

\section{Auxiliary Figure for Section \ref{sec: Monte Carlo}}
\label{sec: aux fig}
\begin{figure}[th!]
  \centering \includegraphics[width =
  0.8\textwidth]{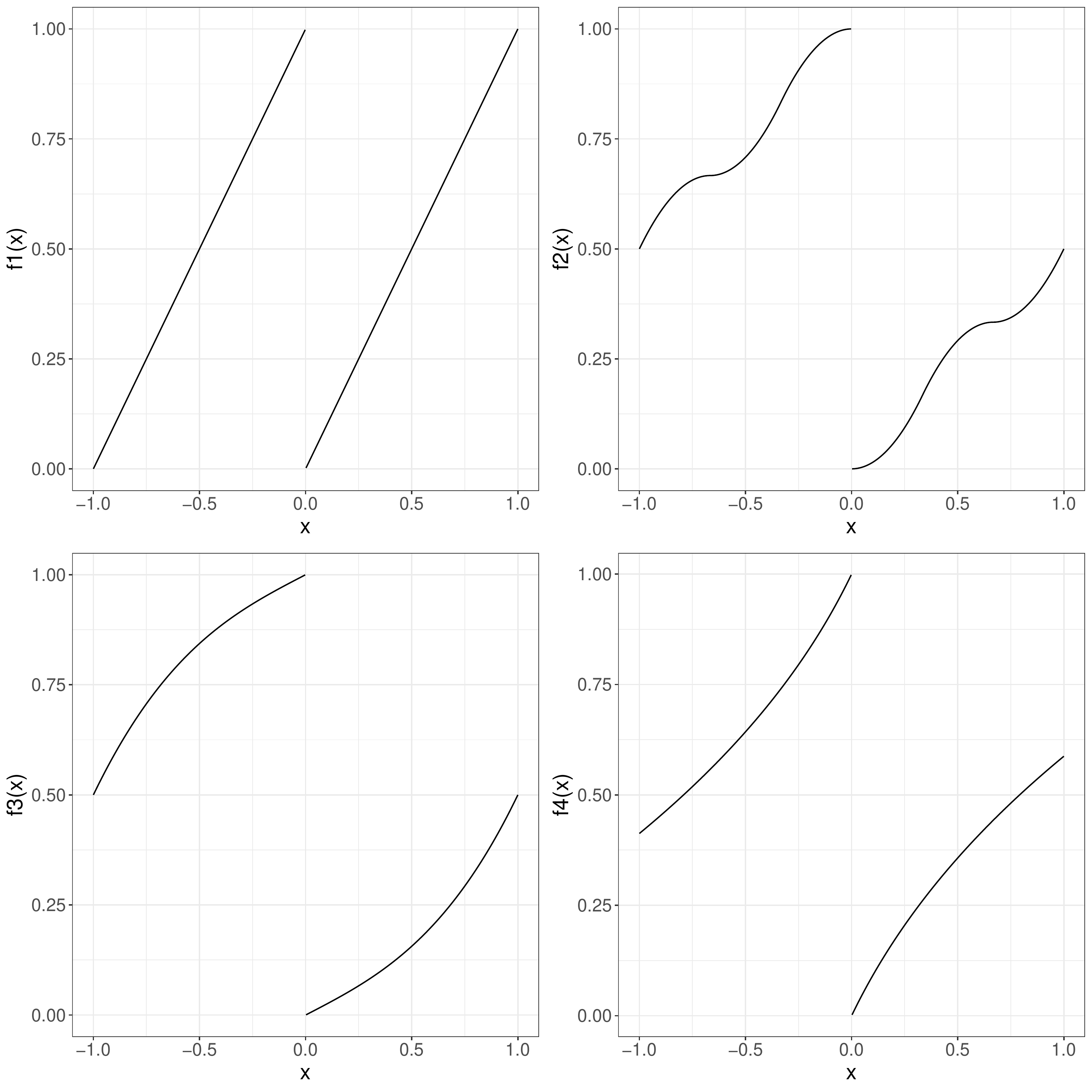}
  \caption{Regression function values on $x \in [-1, 1]$.}
  \label{fig:reg fcns 1dim}
\end{figure}

\end{document}